\documentclass[journal,twoside,web]{ieeecolor}
\usepackage{generic}
\usepackage{amsmath,amssymb,amsfonts}
\usepackage{algorithmic}
\usepackage{graphicx}
\usepackage{algorithm,algorithmic}

\usepackage[
    colorlinks,
    linkcolor=blue,
    anchorcolor=blue,
    citecolor=blue, 
    hyperfigures=true,
    urlcolor=black,
]{hyperref}

\usepackage{textcomp}
\usepackage{physics}
\usepackage{subcaption} 
\let\labelindent\relax
\usepackage[inline]{enumitem}
\usepackage{verbatim}
\newtheorem{theorem}{\textbf{Theorem}}
\newtheorem{corollary}{\textbf{Corollary}}
\newtheorem{lemma}{\textbf{Lemma}}
\newtheorem{remark}{\textbf{Remark}}
\newtheorem{definition}{\textbf{Definition}}
\newtheorem{assumption}{\textbf{Assumption}}
\newtheorem{proposition}{\textbf{Proposition}}
\usepackage{booktabs}

\usepackage{labbatchrev}

\setboolean{showModification}{false}
\setboolean{showRevA}{true}%\rev, \revA
\setboolean{showRevB}{true}
\setboolean{showRevC}{true}
\setboolean{showRevD}{true}
\setboolean{showRevE}{true}
\setboolean{showRevF}{true}
\setboolean{showRevG}{true}

\usepackage{ifthen}
\newboolean{ArxivVersion}
\setboolean{ArxivVersion}{true}

\newcommand{\Cselect}{\mathcal{C}_{\text{select}}}
\newcommand{\Cupdate}{\mathcal{C}_{\text{update}}}

\def\BibTeX{{\rm B\kern-.05em{\sc i\kern-.025em b}\kern-.08em
    T\kern-.1667em\lower.7ex\hbox{E}\kern-.125emX}}
\begin{document}

\title{Online Coreset Selection for Learning Dynamic Systems}

\author{
    Jingyuan Li, 
    Dawei Shi, 
    and Ling Shi
\thanks{J. Li and D. Shi are with the MIIT Key Laboratory of Servo Motion System Drive and Control, School of Automation, Beijing Institute of Technology, Beijing 100081, China (e-mail: nearfar1jy@gmail.com, daweishi@bit.edu.cn).}
\thanks{Ling Shi is with the Department of Electronic and Computer Engineering, Hong Kong University of Science and Technology, Clear Water Bay, Kowloon, Hong Kong (e-mail: eesling@ust.hk).}
}

\maketitle

\begin{abstract}
With the increasing availability of streaming data in dynamic systems, a critical challenge in data-driven modeling for control is how to efficiently select informative data to characterize system dynamics. 
In this work, we develop an online coreset selection method for set-membership identification in the presence of process disturbances,  improving data efficiency while preserving convergence guarantees.
Specifically, we derive a stacked polyhedral representation that over-approximates the feasible parameter set.
Based on this representation, we propose a geometric selection criterion that retains a data point only if it induces a sufficient contraction of the feasible set.
Theoretically, the feasible-set volume is shown to converge to zero almost surely under persistently exciting data and a tight disturbance bound.
\rev{When the disturbance bound is mismatched, an explicit Hausdorff-distance bound is derived to quantify the resulting identification error.
In addition, an upper bound on the expected coreset size is established and extensions to nonlinear systems with linear-in-the-parameter structures and to bounded measurement noise are discussed.}
The effectiveness of the proposed method is demonstrated through comprehensive simulation studies.
\end{abstract}


\section{Introduction}
\label{sec:introduction}
\IEEEPARstart{R}{ecent} years have seen rapid advancements in data-driven control (DDC)~\cite{hou2013from,de2019formulas,van2023informativity}.
Among various methodologies for DDC, a well-established procedure involves first characterizing the set of system models consistent with observed data, often termed the \textit{feasible set}, and subsequently designing a robust controller to satisfy control objectives for all systems within this set~\cite{van2023informativity}.
The former task belongs to set-membership identification (SMI), which aims to learn a nominal set of systems from data~\cite{fogel1982value}.
In online learning processes, a major challenge arises from the fact that dynamic systems continuously generate a huge volume of data. 
Directly applying all raw data for DDC will lead to an increasing computational burden and unlimited storage growth~\cite{karimi2017data}.
Consequently, developing efficient strategies to enhance data efficiency while preserving essential information is crucial for practical deployment of set-based data-driven methodologies in real-time control tasks.

Coreset selection has been developed as a prominent technique to efficiently processing large datasets, particularly in machine learning and data analytics~\cite{feldman2020core,borsos2020coresets,yoon2021online}. % borsos2020coresets,yoon2021online
The objective is to identify a small yet representative subset of the original dataset so that models or parameters learned from this coreset closely approximate those derived from the full dataset.
This approach has demonstrated notable success in accelerating various tasks on static, large-scale datasets, such as 
support vector machines~\cite{tsang2005core}, 
$k$-clustering~\cite{har2005smaller}, 
logistic regression~\cite{huggins2016coresets,munteanu2018coresets}, %huggins2016coresets
Gaussian mixture models~\cite{feldman2011scalable,lucic2018training}  %feldman2011scalable,
and deep learning~\cite{sener2017active,guo2022deepcore}. %guo2022deepcore
Nevertheless, extending coreset selection method to dynamic systems introduces significant challenges primarily due to the inherent online nature, the violation of the independent and identically distributed (i.i.d.) assumption, and the task-specific goals of system identification and control.
Despite these difficulties, recent efforts in DDC have explored intelligent data selection strategies to improve data efficiency for dynamic systems.
For example,\cite{umlauft2019feedback} proposed an event-triggered mechanism for Gaussian processes, where data are selected based on model uncertainty. In\cite{zheng2022nonparameteric}, data selection is triggered when prediction error or distance exceeds a threshold under a Lazily Adapted Constant Kinky Inference model. In~\cite{zheng2024safety}, new data are incorporated when the system state deviates significantly from a reference trajectory, prompting model updates.
However, for fully-parameterized linear systems with bounded disturbance noises, the design of coreset selection strategies that ensure provable performance guarantees remains an open challenge.

To enhance data efficiency in dynamic system learning, SMI provides a natural and powerful framework. 
Instead of deriving point estimates \rev{(e.g., least-squares (LS) methods)}, SMI characterizes the entire set of parameters consistent with observed data and prior disturbance bounds~\cite{schweppe1968recursive,fogel1982value,milanese2013bounding}, offering a principled way to assess the informativeness of new data. 
Several methods for enhancing data efficiency in online learning processes have been explored under this framework.
One line of research is called bounded complexity set update, which constrains the feasible set's complexity -- e.g., via polytopes with fixed hyperplane directions~\cite{veres1999limited}, parallelotopes~\cite{chisci1998block}, or compact over-approximations~\cite{broman1990compact}. 
Another line of works focus on informativeness-driven updates, where the selecting data that are beneficial for reducing uncertainty. For instances, the feasible set is updated only when incoming data sufficiently reduce uncertainty characterized ~\cite{fogel1982value,milanese1991optimal}. 
While these approaches aim to improve data efficiency, a crucial aspect remains underexplored: do such selection schemes guarantee convergence of the feasible set? This question is especially critical when representing feasible set from limited or irregularly sampled data streams.

With the above discussion, the convergence of feasible set in online learning processes emerges as a critical concern.
This form of convergence exhibits two key challenges that distinguish it from standard parameter convergence analysis. 
First, it corresponds to worst-case convergence of the points in a set. 
Second, the convergence is concerned under subsets of entire data.
The analysis on the worst-case convergence under the entire dataset has been explored in many previous works~\cite{bai1998convergence,akccay2004size,lu2021robust,lu2023robust,sasfi2023robust,li2024learning}. 
In \cite{bai1998convergence}, a point-wise asymptotic convergence of the feasible set for discrete-time scalar system was derived under tight noise and persistently exciting regressors. 
\cite{akccay2004size} investigated the convergence of size of feasible set. 
For linear systems, \cite{lu2021robust} derived an asymptotic convergence property under persistent excitation of regressors and tight noise. 
In \cite{li2024learning}, a non-asymptotic convergence of size of feasible set under tight noise and block-martingale small-ball (BMSB) condition was proved. 
The convergence properties under irregular or event-triggered sampling have also attracted increasing attention~\cite{diao2018event,tu2024learning,li2020learning}.
In~\cite{diao2018event}, an event-triggered binary identification algorithm was proposed with guaranteed convergence.
In \cite{tu2024learning}, convergence of parameters is discussed considering multi trajectories using least squares methods, and the convergence rate is derived. 
Similarly, \cite{li2020learning} approached learning from irregularly-sampled time series from a missing data perspective, proposing generative models within an encoder-decoder framework to capture the underlying data distributions.
However, the convergence properties of feasible set under a selected coreset have not been analyzed before and need further exploration.

% Paragraph 5
In this paper, we study the problem of online coreset selection design for learning unknown systems within a set-membership identification framework. 
We focus on linear time-invariant system with bounded additive disturbances, and aim to design a criterion for selecting data in online learning processes. The goal is to construct a coreset that yields a feasible set with significantly fewer data points, while still ensuring convergence guarantees for the learning algorithm. 
To achieve this goal, we first propose an over-approximated polyhedral representation of feasible set for \rev{bounded} disturbance descriptions. 
%Inspired by classical Gr\"unbaum's inequality on mass partition of convex bodies~\cite{grunbaum1960partitions,gardner2006geometric}, we design a selection criterion based on the geometrical relation between new data points and existing uncertainty representation. However, several challenges need to be addressed to enable our design.
Inspired by the work of generalized Gr\"unbaum's inequality~\cite{letwin2024generalization}, we design a selection criterion based on the geometrical relation between new data points and existing uncertainty representation. However, several challenges need to be addressed to enable our design.
First, how can the feasible set be effectively represented using only the coreset data?
Moreover, how to update the polyhedral representation as new data points are selected.
In addition, does the dataset selection criterion based on the geometrical relation guarantees the convergence? How to derive the asymptotic and non-asymptotic expression of size of feasible set.
The main contributions of this work are summarized as follows:
\begin{enumerate}
    \item An online coreset selection method under SMI is developed to efficiently update the feasible parameter set (FPS).
    Specifically, we derive a stacked polyhedral representation that fully characterizes the FPS under bounded disturbances~(\textbf{Proposition~\ref{proposition:representation}}).
    Building on this representation, a threshold-triggered selection mechanism is designed to retain only informative data, and an online learning framework is proposed~(\textbf{Algorithm~\ref{algorithm:OnlineLearning}}).

    \item Convergence guarantee for the data-driven feasible set is established. 
    With persistently exciting data and a tight disturbance bound, the volume of the feasible set is shown to converge to zero almost surely~(\textbf{Theorem~\ref{Theorem:ThresholdTriggerConvergence}}).
    \rev{Moreover, when the disturbance bound is not tight, an explicit Hausdorff-distance bound between the feasible sets induced by the coreset and the full dataset is derived, together with an asymptotic characterization of the limiting feasible set~(\textbf{Theorem~\ref{theorem:dist_alpha0}}).} 

    \item Statistical characterizations of the online data selection protocol are provided.
    In particular, an upper bound on the expected number of selected data points is derived~(\textbf{Theorem~\ref{theorem:expected_bounds}}). \rev{The trade-off between complexity and  performance induced by the selection threshold $\alpha_0$ is then quantified.}
    \rev{The proposed method and its properties are further extended to nonlinear models with linear-in-the-parameters structures~(\textbf{Corollary~\ref{corollary:nonlinear}}) and to settings with bounded measurement noise~(\textbf{Corollary~\ref{corollary:noisyMeasurement}}).}
\end{enumerate}

The rest of the paper is organized as follows. 
Section~\ref{section: preliminaries and problem formulation} introduces preliminaries and problem formulation. 
\rev{The polyhedral representation of the feasible set, the design of the online coreset selection criterion and the overall algorithm is presented in Section~\ref{section:OnlineLearningDesign}.}
The main theoretic results, including the convergence of feasible set constructed from coreset, as well as upper bounds on the coreset cardinality are provided in Section~\ref{section:convergenceAnalysis}.
Extensions, computational complexity, and the resulting trade-off are discussed in Section~\ref{section:ExtensionAndDiscussion}.
Several numerical examples and implementation results are represented and discussed in Section~\ref{section:simulation}, followed by some concluding remarks in Section~\ref{section:conclusions}.

\textit{Notation:}
Let $\mathbb{R},\mathbb{Z},\mathbb{R}_{\geq 0},\mathbb{Z}_{>0}$ denote the set of real numbers, integers, nonnegative real numbers and positive integers respectively. 
For integers $i_1<i_2$, let $[i_1:i_2] \triangleq \{i_1, i_1+1, \dots,i_2\}$ and $[N] \triangleq [1:N]$.
Let $\mathbb{R}^n$ and $\mathbb{R}^{n \times m}$ denote the $n$-dimensional Euclidean space and the space of $n \times m$ real matrices, respectively.
For $x,y \in \mathbb{R}^n$, let $\langle x,y\rangle$ denote the inner product, and $x(i)$ its $i$-th coordinate.
For $x,y \in \mathbb{R}^n$, write $x < y$ if $x(i) < y(i)$ for all $i \in [n]$, and $x \leq y$ analogously.
For $M \in \mathbb{R}^{n \times m}$, let $M(i,j)$, $M(i,:)$, and $M(:,j)$ denote the $(i,j)$-th entry, $i$-th row, and $j$-th column, respectively.
$M_1 \succeq M_2$ denotes that $M_1 - M_2$ is positive semi-definite.
Denote the Kronecker product by $M_1 \otimes M_2$ and the Cartesian product of $m$ copies of $\mathbb{R}^n$ by $(\mathbb{R}^n)^m$. 
For $x \in \mathbb{R}^{n}$, let $\Vert x \Vert_2 = \sqrt{\langle x,x \rangle }$ denote the 2-norm of vector and $\Vert x \Vert_{\infty}$ denote $\infty$-norm of vector.
The multivariate normal distribution is denoted by $\mathcal{N}(\mu,\Sigma)$.
For a finite set $\mathcal{A}$, $|\mathcal{A}|$ denotes its cardinality.
\rev{The convex and conic hulls of the columns of a matrix $K$ are denoted by 
$\text{conv}(K) \triangleq \{K \lambda \mid \lambda \geq 0,\, \sum_i \lambda_i = 1\}$ and $\text{cone}(K) \triangleq \{K \eta \mid \eta \geq 0\}$, respectively.}
\rev{For sets $X,X' \subseteq \mathbb{R}^n$, the Minkowski sum is $X\oplus X' \triangleq \{x+x' \mid x \in X,\, x' \in X'\}$.}
\rev{The distance from a point $u\in \mathbb{R}^n$ to set $X\subset \mathbb{R}^n$ is denoted by $\mathrm{dist}(u,X) = \inf_{x \in X} \left\lVert u - x\right\rVert_2$}.

\section{Preliminaries and Problem Formulation}
\label{section: preliminaries and problem formulation}
Consider a discrete linear time-invariant system
\begin{equation} \label{eq:system}
    x_{k+1} = Ax_k + B u_k + w_k,
\end{equation}
where 
$x_k \in \mathbb{R}^{n_x}$ is  the state of the system at time $k$, $u_k \in \mathbb{R}^{n_u}$ is the control input, $w_k \in \mathbb{R}^{n_x}$ is the bounded disturbance, $A \in \mathbb{R}^{n_x \times n_x}$ and $B \in \mathbb{R}^{n_x \times n_u}$ are \textit{unknown}  system matrices.

Let $z_k \triangleq \begin{bmatrix}x_k^T & u_k^T\end{bmatrix}^T \in \mathbb{R}^{n_z}$ be a state-input vector at time $k$ where $n_z = n_x + n_u$. 
For any $K>0$ and any set $S = \{k_1,...,k_{\vert S\vert}\} \subseteq [K+1]$, we denote $Z_S \triangleq [z_{k_1}, \dots, z_{k_{\vert S\vert}}]$ and $X_S \triangleq [x_{k_1}, \dots, x_{k_{\vert S\vert}}]$ as grouped state-input matrix and state matrix indexed by $S$, respectively. 

\subsection{Preliminaries}

\subsubsection{Set-membership Identification}
In this section, we review SMI for linear system with bounded noise.  
Consider the system described by equation~(\ref{eq:system}). 
Suppose that the additive disturbance $w_k$ satisfies the following assumption:

\begin{assumption}
    \label{as:disturbance}
    The disturbance sequence $ \{w_k\}_{k \geq 0} \subset \mathbb{R}^{n_x}$ takes values in a known compact, convex set $ \mathcal{W} \subset \mathbb{R}^{n_x} $, which is assumed symmetric w.r.t. the origin.
    Moreover, $ \{w_k\}_{k \geq 0} $ is assumed to be a sequence of independent and identically distributed (i.i.d.) random vectors with a probability measure $ \mathbb{P} $ supported on $ \mathcal{W} $.
\end{assumption}

\begin{remark}
Several existing works impose specific structures on the disturbance set $\mathcal{W}$, e.g., a bounded polytope \cite{lorenzen2019robust}, an ellipsoid \cite{feng2025data}, or a zonotope \cite{alanwar2023data}. Here, we only assume that $ \mathcal{W} \subset \mathbb{R}^{n_x}$ is a known compact and convex set, which covers a broader class of uncertainty descriptions.
In addition, the origin symmetry of $\mathcal{W}$ is adopted mainly for convenience; no symmetry is imposed on $\mathbb{P}$, hence $w_k$ need not be zero-mean.
\end{remark}

Let $\Theta_{*} = \begin{bmatrix} A & B \end{bmatrix}$ denote the true system parameters. 
For any subset $\mathcal{S}  \subseteq \mathbb{Z}_{>0} $, we define
\begin{equation}
    \label{eq:SME}
    \Theta_{\mathcal{S}} \triangleq \bigcap_{k \in \mathcal{S}} \{{\Theta}:x_{k}-{\Theta} z_{k-1} \in \mathcal{W}\}.
\end{equation}
Then, by construction,  the true system parameters $\Theta_*$ satisfies:
\begin{equation} \label{eq:SME_Theta*}
    \Theta_* \in \Theta_{\mathcal{S}}.
\end{equation}

The intersection $\Theta_{\mathcal{S}}$ provides a set-valued characterization of the uncertainty in the system parameters, which can be viewed as an uncertainty qualification for  $\Theta_*$. 
The right-hand side of~(\ref{eq:SME_Theta*}) contains all feasible parameters consistent with the input-output data.

\subsubsection{Polyhedra and Polytopes}
We introduce some basic concepts and properties related to polyhedral representation in this section.
Let $H \in \mathbb{R}^{n \times m}$ be a matrix and $c \in \mathbb{R}^n$ be a vector. A polyhedron in $\mathcal{H}$-representation (also called an $\mathcal{H}$-polyhedron) is defined by
\begin{equation}\label{eq:H-polyhedron}
    P_{\mathcal{H}}(H, c) \triangleq \{ x \in \mathbb{R}^m \mid Hx \leq c \},
\end{equation}
which is the solution set of a system of linear inequalities.
By the Minkowski-Weyl theorem~\cite{minkowski1989allgemeine}, any polyhedron admits a dual $\mathcal{V}$-representation, characterized as the Minkowski sum of the convex hull of a finite set of vertices $V=\{v_1,\dots,v_p\}$ and the conic hull of a finite set of rays $R=\{r_1,\dots,r_q\}$:
\begin{equation}\label{eq:V-polyhedron}
    P_{\mathcal{V}}(V, R) \triangleq \text{conv}(V) \oplus \text{cone}(R),
\end{equation}
where $\text{conv}(\cdot)$ and $\text{cone}(\cdot)$ are defined in \textit{Notation}.

The centroid of a polytope $P$ is defined as 
\begin{equation}
    \label{eq:centroid}
    g(P) \triangleq \frac{1}{\mu(P)} \int_P x \dd{x},
\end{equation}
where $\mu(\cdot)$ denotes the Lebesgue measure (volume) of an $m$-dimensional convex body.

The \textit{supporting function} $h_P: \mathbb{R}^m \to \mathbb{R}$ for a convex set $P \subset \mathbb{R}^m$ is
\begin{equation}
    h_P(\xi) \triangleq \max_{x \in P}\{\langle x, \xi\rangle\}.
\end{equation} 

We denote the \textit{centered supporting function} $\widetilde{h}_P:\mathbb{R}^m \to \mathbb{R}$ w.r.t. the centroid $g(P)$ as :
\begin{equation}\label{eq:centeredSF}
    \widetilde{h}_P(\xi) \triangleq \max_{ x \in P}\{\langle x - g(P), \xi\rangle\}.
\end{equation}
Then, by \cite{weyl1934elementare}, we have the following inequalities:
\begin{equation}
    \label{eq:h_p(pmZ)}
    \frac{1}{m} \widetilde{h}_P(\xi) \leq \widetilde{h}_P(-\xi) \leq m \widetilde{h}_P(\xi).
\end{equation}

\subsection{Problem Formulation}
Due to uncertainty of $w$, we cannot obtain $\Theta_*$ precisely from dataset. 
Alternatively, we leverage set-membership identification to characterize the feasible set of $\Theta_*$ and design selection mechanism in an online fashion to enhance data efficiency. 
Over time period $[K]$, if all process data are kept as coreset (the trivial case), the SMI of $\Theta_*$ by~(\ref{eq:SME}) is represented by $\Theta_{[K]}$.
Denote $\mathcal{D}_k = \{z_{k-1},x_{k}\}$ as the dataset to be determined to select at time step $k$ which determines a feasible set as shown in~(\ref{eq:SME}). Let $\mathcal{S}_K \subset [K]$ be a subset of all time steps until $K$. 
Then, the coreset constructed for SMI is represented by $\mathcal{D}_{\mathcal{S}_K} = \cup_{k \in \mathcal{S}_K} \mathcal{D}_k$.
The corresponding SMI for $\Theta_*$ is denoted by $\Theta_{\mathcal{S}_K}$.

Our objective is to design an online coreset selection criterion to update  $\mathcal{S}_K$. 
The selected coreset yields a reduced dataset whose size is significantly smaller than total number of samples $K$, while ensuring some learning performance guarantees. 
However, designing such a selection criterion is nontrivial due to several intertwined challenges:

First, the basic \rev{SMI} formulation (\ref{eq:SME}) does not explicitly characterize the \textit{estimation uncertainty}, i.e., it does not reveal how the volume or size of the feasible set evolves as new data (e.g., $x_{k+1}$) arrive. 
An analytical representation of the feasible set is required for such $\mathcal{W}$. This motivates the first problem:

\textbf{Problem 1}: How can we construct a tractable and analytical representation of the feasible set that enables quantification and comparison of estimation uncertainty?

Given such a representation, the next challenge is to develop a principled efficient data selection criterion that determines whether newly arriving data should be retained.

\textbf{Problem 2}: How to design a selection criterion that retains informative data points for reducing uncertainty of feasible set?

Finally, and most critically, the designed selection criterion should preserve the theoretical guarantees of the SMI framework. 
In particular, it is required that the feasible set constructed in an online fashion using coreset converges to true parameters as more data become available. 
This leads to the third problem:

\textbf{Problem 3}: Under the proposed coreset selection criterion, how to establish convergence properties of the feasible set? 
What are the statistical properties of the selection probability and the cumulative number of selected data points in online learning processes?

\section{Online Coreset Selection Design}
\label{section:OnlineLearningDesign}

In this section we address Problems 1 and 2. 
First, we derive a stacked polyhedral over-approximation of feasible parameters set, which provides a tractable and analytical representation of the estimation uncertainty. 
Subsequently, we propose a coreset selection  mechanism that identifies and retains informative data points.
Finally, based on the double description method for polyhedra, we develop a dynamic constraints reduction algorithm to further optimize the representation by eliminating redundant constraints.

\subsection{Polyhedral Over-approximation of Feasible Set}
\label{subsection:representation}
The main idea behind deriving the data-based representation of full parameterized feasible set is to decompose uncertain matrix $\Theta \subset \mathbb{R}^{n_x\times n_z}$ into $n_x$ row vectors in $\mathbb{R}^{n_z}$, each constrained within a data-dependent polyhedron. 
To illustrate the idea, we denote row-stacking map $\Phi : (\mathbb{R}^n)^m \to \mathbb{R}^{m \times n}$ as 
\begin{equation}
    \label{eq:Phi}
    \Phi(v_1,...,v_m) = \begin{bmatrix}v_1^T & \dots & v_m^T\end{bmatrix}^T,
\end{equation}
where each $v_i \subset \mathbb{R}^n$ is treated as a column vector set, and $v_i^T$ denotes its transpose(a row vector set). 
Then, matrix set $\Theta \subset \mathbb{R}^{n_x\times n_z}$ can be represented as $\Theta = \Phi(\theta^1,\dots,\theta^{n_x})$ where $\theta^i \subset \mathbb{R}^{n_z}$ for all $i \in [n_x]$. 
The true system parameters $\Theta_*$ can be written as $\Theta_* = \Phi(\theta_*^1, \theta_*^2, \dots, \theta_*^{n_x})$ where $\theta_*^i \in \mathbb{R}^{n_z}$ are constant vectors.   
Denote initial uncertainty matrix set $\Theta_0$ as 
\begin{equation}
    \label{eq:Theta0}
    \Theta_0 \triangleq \Phi(\theta_0^1,\dots,\theta_0^{n_x}).
\end{equation} We consider the following assumptions on $\Theta_0$:
\begin{assumption}\label{as:initialSet}
    $\forall i \in [n_x]$, $\theta_*^i $ lies in a known, bounded, and convex polytope $\theta_0^i = P_{\mathcal{H}}(H_0^i,c_0^i)$, i.e., $\theta_*^i \in \theta_0^i = \{\theta^i : H_0^i \theta^i \leq c_0^i\}$. 
\end{assumption}

Note that assumptions on prior bounded uncertainty are common in set-membership identification~\cite{lorenzen2019robust}.  
\rev{To ensure tractable set operations, let $\bar{w}\in \mathbb{R}^{n_x}$ denote the vector of element-wise \revE{maximum} of $\mathcal{W}$. Under the symmetry assumption in Assumption~\ref{as:disturbance}, it follows that $w \in \mathcal{W} \implies -\bar{w} \leq w \leq \bar{w}$}.
\rev{This amplitude-bounded form serves as a conservative outer approximation of the original disturbance set.}

\begin{revBlock}
We now derive the data-driven representation of the feasible set.
\revE{The construction is separable across state equations: each component $i\in[n_x]$ induces a feasible set in $\mathbb{R}^{n_z}$, and the feasible set for the full system is obtained by stacking these component-wise sets via the row-stacking map $\Phi(\cdot)$ in~\eqref{eq:Phi}.}
\revE{To reduce notation burden without loss of generality, we first present the construction for a single state equation (a fixed $i\in[n_x]$).}
At each time step $k$, we have $x_{k}(i) = z_{k-1}^T \theta_{*}^i + w_{k-1}(i)$ and $w_{k-1}(i) \in [-\bar{w}(i), \bar{w}(i)]$. From Eq.~(\ref{eq:SME}), the true parameter $\theta_*^i$ must lie within the feasible \revE{set} $\mathcal{S}_k^i$:
\begin{equation}\label{eq:one_state_illustration}
    \mathcal{S}_k^i \triangleq \{ \theta^i  : -\bar{w}(i) + x_{k}(i) \leq z_{k-1}^T \theta^i \leq \bar{w}(i) + x_{k}(i)\}.
\end{equation}

Over a horizon of $K$ samples, the feasible parameter set for the $i$-th component is the intersection of these sets with the initial set $\theta_0^i$. This intersection naturally forms a polyhedron:
\begin{equation}\label{eq:intersection_logic}
    P(\mathcal{D}_{[K]}^i) = \theta_0^i \cap \left( \bigcap_{k=1}^K \mathcal{S}_k^i \right),
\end{equation}
\revE{which} can be written in the following compact $\mathcal{H}$-representation:
\begin{equation}\label{eq:PolyhedralSet}
    P(\mathcal{D}_{[K]}^i) \triangleq \{ \theta \in \mathbb{R}^{n_z}: H_K^i \theta \leq C_K^i \},
\end{equation}
where the matrices $H_K^i$ and $C_K^i$ are constructed by stacking the constraints from \eqref{eq:one_state_illustration}:
\begin{equation}
    \label{eq:HC}
    \!\!\!H_K^i \triangleq \begin{bmatrix}
        H_0^i \\ 
    Z_{[0:K-1]}^T \\
    -Z_{[0:K-1]}^T
    \end{bmatrix}\!,
    C_K^i \triangleq \begin{bmatrix}
        c_0^i \\ 
        (X_{[K]}(i,:))^T + \overline{W}(i)  \\ 
        -(X_{[K]}(i,:))^T + \overline{W}(i)
    \end{bmatrix}\!,
\end{equation}
with $\overline{W}(i) \triangleq \mathbf{1}_{K} \otimes \bar{w}(i)$ where $\mathbf{1}_K \triangleq [1,\dots,1]^T \in \mathbb{R}^K$ is the all-ones vector. 

\revE{Since $i$ is arbitrary, the above construction applies to every state equation $i\in[n_x]$.
We then consider the full system and define the stacked feasible set via $\Phi(\cdot)$ in~\eqref{eq:Phi}, which is summarized in the following proposition.}

\end{revBlock}

\begin{proposition}
    \label{proposition:representation}
    Consider system (\ref{eq:system}) with full dataset $\mathcal{D}_{[K]}$. Suppose Assumptions \ref{as:disturbance} and \ref{as:initialSet} hold. Then, the stacked polyhedral set
    \begin{equation}
        \label{eq:polyhedral_representation}
        \widehat{\Theta }_{K}^{\text{full}} \triangleq \Phi(P(\mathcal{D}_{[K]}^1),P(\mathcal{D}_{[K]}^2),...,P(\mathcal{D}_{[K]}^{n_x}))
    \end{equation}
    satisfies
    \begin{equation}
        \widehat{\Theta }_{K}^{\text{full}} \supseteq (\Theta_0 \cap \Theta_{[K]}),
    \end{equation}
    where $\Theta_0$ and $\Theta_K$ are defined in~(\ref{eq:Theta0}) and (\ref{eq:SME_Theta*}), respectively.
\end{proposition}

\begin{proof}
    We complete the proof by showing that $\widehat{\Theta }_{K}^{\text{full}} $ is an over-approximation of $\Theta_0 \cap \Theta_{[K]}$. 
    For all $w \in \mathcal{W}$, every coordinate of $w$ is bounded by $\bar{w}$. Then, we have:
    \begin{equation}\label{eq:bound}
        \begin{bmatrix}
        I \\
        -I
        \end{bmatrix}w\leq 
        \begin{bmatrix}
        \bar{w} \\
        \bar{w}
        \end{bmatrix}.
    \end{equation}

    Based on \rev{SMI}, for each $\Theta = \Phi(\theta^1,\dots,\theta^{n_x}) \in \Theta_0 \cap \Theta_{[K]}$, we have
    $x_{k} - \Theta z_{k-1} \in \mathcal{W}$. Thus, 
    \begin{equation}\label{eq:boundState}
        \begin{bmatrix}
            I \\
            -I
        \end{bmatrix}
        (x_{k} - \Theta z_{k-1}) \leq
        \begin{bmatrix}
            \bar{w} \\
            \bar{w}
        \end{bmatrix} 
    \end{equation} holds for all $k \in [K]$. For all $i \in [n_x]$, stacking the $i$-th and the $(i+n_x)$-th rows horizontally and transposing the resulting matrix, we infer from~(\ref{eq:boundState}) that
    \begin{equation}
        \begin{bmatrix}
        z_{k-1}^T \\
        -z_{k-1}^T
        \end{bmatrix}\theta^i \leq
        \begin{bmatrix}
        x_{k}(i) + \bar{w}(i) \\
        -x_{k}(i) + \bar{w}(i)
        \end{bmatrix}.
    \end{equation}
    Then, by incorporating the initial feasible constraints of $\theta^i$ in Assumption~\ref{as:initialSet}, we have $\theta^i \in P(\mathcal{D}_K^i)$ holds for all $i \in [n_x]$. Consequently, $\Theta \in \widehat{\Theta }_{K}^{\text{full}}$, and thus $ \widehat{\Theta }_{K}^{\text{full}} \supseteq (\Theta_0 \cap \Theta_{[K]})$.
\end{proof}

\begin{remark}
    Proposition~\ref{proposition:representation} provides an explicit data-based representation $\widehat{\Theta}_K^{\text{full}}$ of the true parameters $\Theta_*$ by \rev{SMI} using the full dataset over time period $[K]$.
    Compared with existing polyhedral representation of feasible set that directly operates on $\mathrm{vec}(\Theta_*) \in \mathbb{R}^{n_x n_z}$(i.e. \cite{lauricella2020set}), the proposed representation has a key advantage: each polyhedral component of $\widehat{\Theta}_K^{\text{full}}$ lies in a significantly lower-dimensional space $\mathbb{R}^{n_z}$. 
    This dimensionality reduction enables the use of standard polyhedral operations--such as conversions between vertex and hyperplane-representations--for systems of moderate size, which could reduce computational burden commonly encountered in high-dimensional polyhedral set representations.
    By this representation, one can achieve a favorable trade-off: although the number of polyhedra increases, their individual low dimension renders many otherwise intractable operations computationally feasible.
\end{remark}

\subsection{Criterion for Coreset Selection}
\label{subsection:criterion}
In this subsection we focus on the problem of how to design coreset selection criterion for data selection and thus improve data efficiency.
Suppose the learning process has continued up to time step $K-1$. 
The set $\Phi(P_{K-1}^1,P_{K-1}^2,\dots,P_{K-1}^{n_x})$ is a data-based representation containing true parameters $\Theta_*$ where $P_{K-1}^i$ for $i \in [n_x]$ are polyhedral sets constructed from selected process data. 
At time instant $K$, for each $i \in [n_x]$, the sampling of $x_{K}(i)$ generates two half-space constraints which can be written as 
\begin{align} 
    H_{+}^i & := \{\theta: \quad (z_{K-1})^T \theta \leq x_{K}(i) + \bar{w}(i)\},\label{eq:H+} \\
    H_{-}^i &:= \{\theta: (-z_{K-1})^T \theta \leq -x_{K}(i) + \bar{w}(i)\} .\label{eq:H-}
\end{align}

These half-space constraints can be equivalently expressed via centered supporting function. 
For the half-space constraints $H_{+}^i$ in~(\ref{eq:H+}), it can be reformulated as:
\begin{equation}\label{eq:z1}
    -z_{K-1}^T (\theta - g(P_{K-1}^i)) \geq -x_{K}(i) - \bar{w}(i) + z_{K-1}^T g(P_{K-1}^i).
\end{equation}
The right-hand side of (\ref{eq:z1}) can be rewritten as $\alpha_+^i \widetilde{h}_{P_{K-1}^i}(z_{K-1})$. 
Here, $\widetilde{h}_{P_{K-1}^i}(z_{K-1})$ denotes the centered supporting function defined in (\ref{eq:centeredSF}), and the scalar $\alpha_+^i$, \rev{which we refer to as the \textit{normalized offset}}, is given by
\begin{equation}\label{eq:alpha+}
    \alpha_+^i = \frac{-x_{K}(i) - \bar{w}(i)+z_{K-1}^T g(P_{K-1}^i)}{\widetilde{h}_{P_{K-1}^i}(z_{K-1})}.
\end{equation}

We can apply the same procedure to $H_{-}^i$ and obtain the analogous expression:
\begin{equation}
    H_{-}^i = 
    \{
        \theta: z_{K-1}^T (\theta - g(P_{K-1}^i)) \geq \alpha_-^i \widetilde{h}_{P_{K-1}^i}(-z_{K-1})
    \}
\end{equation}
with the corresponding normalized offset:
\begin{equation}\label{eq:alpha-}
    \alpha_-^i = \frac{x_{K}(i) - \bar{w}(i)-z_{K-1}^T g(P_{K-1}^i)}{\widetilde{h}_{P_{K-1}^i}(-z_{K-1})}.
\end{equation}

Then, we have the following proposition in a compact form.

\begin{proposition}
    \label{proposition:SupportingFunctionRepresentation}
    Let $K \in \mathbb{Z}_{>0}$ be fixed. For each $i \in [n_x]$, the constraints $H_{+}^i$ and $H_{-}^i$ in (\ref{eq:H+}) and (\ref{eq:H-}) admit the equivalent representation in terms of $\alpha_{+}^i$ and $\alpha_{-}^i$ in the following form:
    \begin{align*}
        H_{+}^i &= \{\theta : \langle -z_{K-1}, \theta - g(P_{K-1}^i) \rangle \geq \alpha_{+}^i \widetilde{h}_{P_{K-1}^i}(z_{K-1}) \}, \\ 
        H_{-}^i &= \{\theta : \langle z_{K-1}, \theta - g(P_{K-1}^i) \rangle \geq \alpha_{-}^i \widetilde{h}_{P_{K-1}^i}(- z_{K-1}) \},
    \end{align*}
    where $\alpha_{+}^i$ and $\alpha_{-}^i$ are given in (\ref{eq:alpha+}) and (\ref{eq:alpha-}) respectively.
\end{proposition}

\begin{remark}
    \rev{The parameters $\alpha_{+}^i$ and $\alpha_{-}^i$ characterize the relative location of the induced hyperplanes with respect to the polyhedral uncertainty set.}
    Consider half-space $H_{+}^i$ and the polytope $P_{K-1}^i$ as an example:
    When $\alpha_{+}^i = -1$, (\ref{eq:H+}) can be reformulated as $z_{K-1}^T(\theta - g(P_{K-1}^i)) \leq \widetilde{h}_{P_{K-1}^i}(z_{K-1})$. 
    This inequality holds for all $\theta \in P_{K-1}^i$. 
    Equality is achieved when $\theta = \arg \max_{x \in P_{K-1}^i} \{ \langle x - g(P_{K-1}^i),z_{K-1}\rangle \}$, indicating that $H_{+}$ is tangent to $P_{K-1}^i$ at its supporting hyperplane. 
    When $\alpha_{+}^i = 0$,  (\ref{eq:H+}) can be reformulated as $z_{K-1}^T(\theta - g(P_{K-1}^i)) \leq 0$, indicating that the hyperplane $\{\theta: z_{K-1}^T \theta = x_{K}(i) + \bar{w}(i)\}$ passes exactly through the centroid $g(P_{K-1}^i)$.  
    For $\alpha_{+}^i \in (-1,0)$, it's guaranteed that $H_{+}^i \cap P_{K-1}^i \neq \emptyset$ and $ H_{+}^i \cap P_{K-1}^i \subset P_{K-1}^i$, which shows that adding ${z_{K-1}, x_{K}^i}$ at time instant $K$ can decrease size of feasible set.
\end{remark}

Based on the above discussion, we propose the following selection criterion for constructing the coreset.

Let $\alpha_0 \in [-1,0)$ be a predefined threshold. 
For each dimension $i \in [n_x]$ and time instant $K$, define the sub-trigger $\gamma_{K,i}\in \{0,1\}$ as 
\begin{equation}\label{eq:subtrigger}
    \gamma_{K,i} =1 \iff  \max(\alpha_+^i, \alpha_-^i) \geq \alpha_0,
\end{equation}
where $\alpha_{+}^i,\alpha_{-}^i$ are determined in (\ref{eq:alpha+}) and (\ref{eq:alpha-}) respectively. 
Let $\vee$ denotes the logical OR, the binary trigger $\gamma_{K}\in \{0,1\}$ is defined as :
\begin{equation}\label{eq:trigger}
    \gamma_K = 1 \iff \vee_{i=1}^{n_x} \gamma_{K,i} = 1.
\end{equation}
The coreset is constructed recursively according to 
\begin{equation}
    \label{eq:triggerData}
    \mathcal{D}_{K} =  
    \begin{cases}
        \mathcal{D}_{K-1} \cup \{z_{K-1},x_{K}\} & \text{if }\gamma_K = 1,\\
        \mathcal{D}_{K-1} & \text{if }\gamma_K = 0.
    \end{cases}
\end{equation}

The trigger $\gamma_K = 1$ implies that there exists $i \in [n_x]$ such that a newly constructed half-space constraints determined by $z_{K-1}$ and $x_{K}^i$ has a \rev{normalized offset} $\alpha^i$, with respect to $P_{K-1}^i$, which is larger than $\alpha_0$. 
\rev{The choice of $\alpha_0$ determines the strictness of the data selection.}
\rev{The boundary case $\alpha_0 = -1$ corresponds to a baseline strategy where any constraint intersecting the current feasible set triggers an update.}
\rev{In contrast, setting $\alpha_0 \in (-1, 0)$ selects only data that significantly tighten the feasible set.}

\subsection{Online Learning Framework and Implementation}
\label{sec:online_learning_framework}
\begin{revBlock}
In this section we introduce the complete online learning procedure. 
The schematic diagram is presented in Fig.~\ref{fig:schematic_diagram}. 
\end{revBlock}

\begin{figure}[htbp]
    \centering
    \includegraphics[width=.45\textwidth]{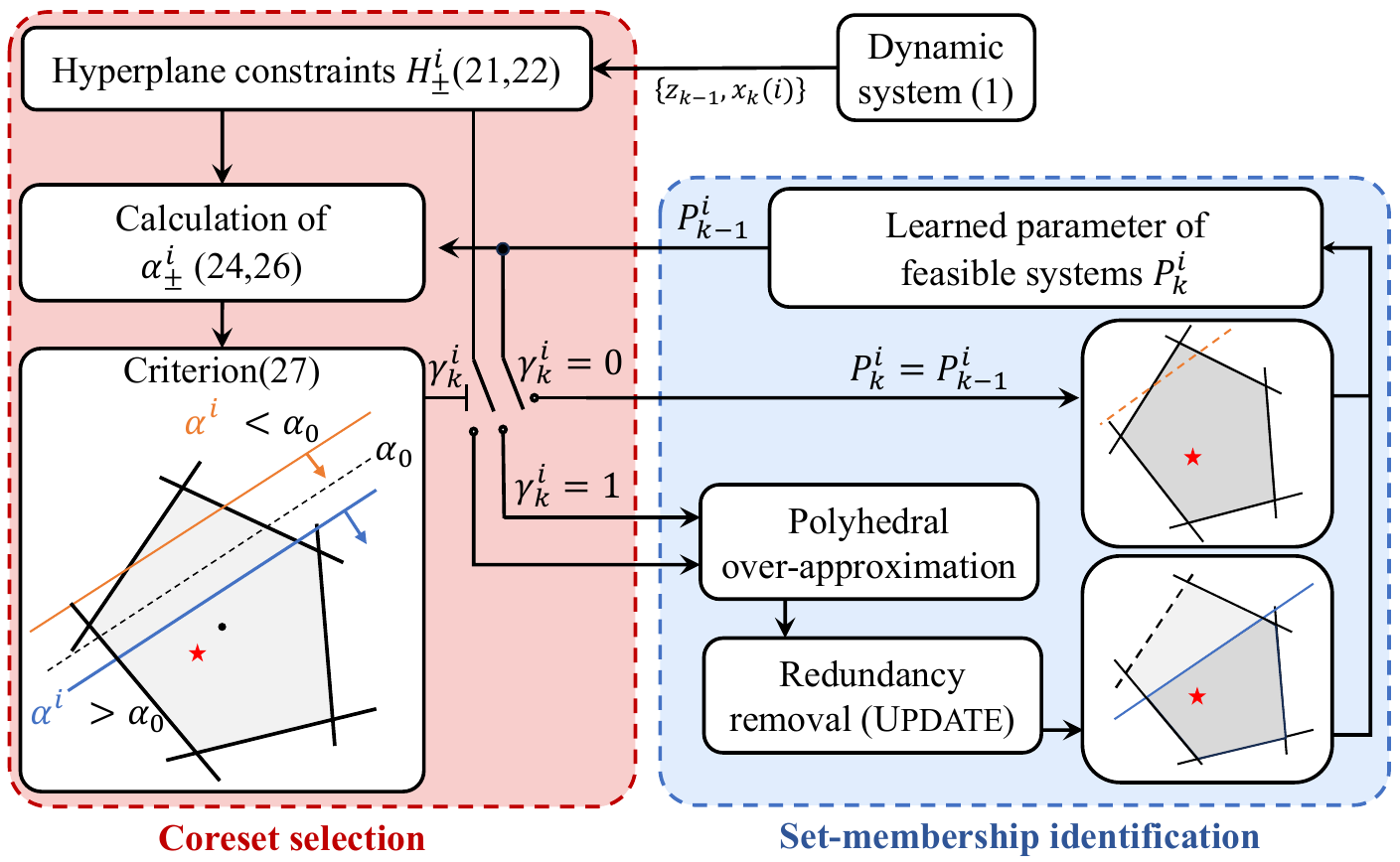}%
    \caption{Schematic diagram of the proposed online learning framework.}
    \label{fig:schematic_diagram}
\end{figure}

\begin{revBlockB}
Based on the proposed polyhedral representation in Section~\ref{subsection:representation} and selection mechanism in Section~\ref{subsection:criterion}, the overall online learning can be designed and summarized in Algorithm~\ref{algorithm:OnlineLearning}.
At time $k$, for the $i$-th polyhedron, we first construct the candidate half-spaces $H_\pm^i$ from the incoming sample $\{z_{k-1},x_k(i)\}$ via~\eqref{eq:H+}--\eqref{eq:H-}~(Line~\ref{line:Hpm}).
We then compute the normalized offsets $\alpha_\pm^i$ according to~\eqref{eq:alpha+} and \eqref{eq:alpha-} (Line~\ref{line:alpha}), which quantify how $H_\pm^i$ would cut $P_{k-1}^i$.
The selection criterion~\eqref{eq:subtrigger} determines $\gamma_{k,i}$, and the global trigger $\gamma_k$ is obtained via~\eqref{eq:trigger}~(Line~\ref{line:trigger}).
If $\gamma_k=1$, only the polyhedra with $\gamma_{k,i} = 1$ are updated by intersecting with $H_+^i$ and $H_-^i$~(Line~\ref{line:append}).
The representation is then updated by removing redundant constraints and recomputing the centroid, denoted by \textsc{Update}$(\cdot)$ in Line~\ref{line:reduce}.
Otherwise $P_{k}^i = P_{k-1}^i$ and the feasible set remains unchanged (Line~\ref{line:equiv}).
Finally, the stacked feasible set $\widehat{\Theta}_k = \Phi(P_k^1,\dots,P_k^{n_x})$
is formed from the row-wise polyhedra.
\end{revBlockB}

\begin{algorithm} 
    \renewcommand{\algorithmicrequire}{\textbf{Input:}}
    \renewcommand{\algorithmicensure}{\textbf{Output:}}
    
    \caption{Online Learning}    \label{algorithm:OnlineLearning}
    \begin{algorithmic}[1]
    \REQUIRE Initial polyhedral representations $\{P_0^i\}_{i=1}^{n_x}$ and triggering threshold $\alpha_0$.
    \ENSURE Feasible set $\widehat{\Theta}_k$ for each $k$.
    \FOR{$k=1,2,\dots$}
        \FOR{$i = 1,\dots,n_x$}
            \STATE Form half-spaces $H_{+}^i$ and $H_{-}^i$ via \eqref{eq:H+} and \eqref{eq:H-}.   \label{line:Hpm}      
            \STATE Compute offsets $\alpha_{\pm}^i$ via (\ref{eq:alpha+}) and (\ref{eq:alpha-}) using $P_{k-1}^i$ and  $g(P_{k-1}^i)$.\label{line:alpha}
            \STATE Determine $\gamma_{k,i}$ by (\ref{eq:subtrigger}) with triggering threshold $\alpha_0$.\label{line:sub_trigger}
        \ENDFOR
        \STATE Calculate $\gamma_k$ by (\ref{eq:trigger}) and get $\mathcal{D}_k$ by (\ref{eq:triggerData}).\label{line:trigger}
        \IF{$\gamma_k = 1$}
            \FOR{$i = 1,\dots,n_x$}
                \IF{$\gamma_{k,i} = 1$}
                    \STATE $\overline{P}_k^i \gets P_{k-1}^i \cap H_{+}^i \cap H_{-}^i$ \label{line:append}
                    \STATE $P_k^i,g(P_k^i) \gets \textsc{Update}(\overline{P}_k^i)$ \label{line:reduce}% \COMMENT{Apply DD or LP-based update}
                \ELSE
                    \STATE $P_k^i \gets P_{k-1}^i$ \label{line:equiv}
                \ENDIF
            \ENDFOR 
            \STATE $\widehat{\Theta}_k \gets \Phi(P_{k}^1, \dots, P_{k}^{n_x})$.
        \ELSE 
            \STATE $\widehat{\Theta}_k \gets \widehat{\Theta}_{k-1}$.
        \ENDIF
    \ENDFOR
    \end{algorithmic}
\end{algorithm}

\begin{revBlockE}
\subsubsection*{Implementation of \textsc{Update}}
\label{subsection:reduction}
The practical realization of the \textsc{Update} function is determined by the tractability of vertex enumeration. 
When the enumeration of vertices is tractable~\cite{avis1995good}, the Double Description (DD) method~\cite{fukuda1995double} can be used.
By maintaining both $\mathcal{V}$ and $\mathcal{H}$ representations, redundancy can be efficiently identified by checking the adjacency of vertices and the existence of active points on each hyperplane. 
Additionally, the centroid $g(P_k^i)$ can be obtained exactly via standard simplex decomposition, with cost dominated by the DD update.

Conversely, when vertex enumeration becomes intractable, redundancy can be identified via  linear programming (LP).
Let $h_0 = \{x : h^T x \leq c\}$ be a constraint in polyhedral representation $P_{\mathcal{H}}(H_0,C_0)$, and let $\bar{P}$ denote the polyhedron formed by removing $h_0$. 
Then, the constraint $h_0$ is a redundant constraint if and only if: $\min_{x \in \bar{P}} (-h^T x + c) \geq 0$.
Moreover, the centroid $g(P_k^i)$ can be approximated using standard randomized schemes, such as the Hit-and-Run sampler~\cite{gilks2001following,del2006sequential}. 
This approximation introduces an error $\Delta g$, resulting in a deviation of $\alpha_{\pm}^i$ bounded by $\vert \Delta \alpha\vert \leq r_k \left\lVert \Delta g\right\rVert_2  + o(\left\lVert \Delta g\right\rVert_2)$, where $r_k = \left\lVert z_{k-1}\right\rVert_2/\widetilde{h}_{P_{k-1}^i}(z_{k-1})$.  
Crucially, as this perturbation affects only the triggering decision rather than the parameter update itself, the consistency of the feasible set is preserved.
Consequently, the selection mechanism under approximation error can be interpreted as operating with a perturbed threshold $\alpha_0'=\alpha_0-\Delta\alpha$, ensuring that all selected half-spaces remain sufficiently informative relative to this bound.
\end{revBlockE}

\section{Convergence Analysis}
\label{section:convergenceAnalysis}
\rev{This section establishes the convergence guarantees and statistical properties of the proposed online coreset selection method under following conditions on the regressor sequence and the disturbance.}

\begin{definition}[Persistent excitation]\label{def:PE}
A sequence $\{z_k\}_{k\ge 0}$ is said to be
\emph{$(\beta,N_u,b_z)$-persistently exciting} (abbrev. $(\beta,N_u,b_z)$-PE) if there exist
constants $\beta>0$, $N_u\in\mathbb{Z}_{>0}$, and $b_z>0$ such that $\|z_k\|_2 \le b_z$, and for all $k_0\ge 0$,
\begin{equation}\label{eq:PE}
\frac{1}{N_u}\sum_{k=k_0+1}^{k_0+N_u} z_k z_k^\top \succeq \beta^2 I.
\end{equation}
\end{definition}

\begin{definition}[Tightness of the bound]\label{def:tightness}
Let $\xi$ be a scalar random variable supported on $[-\bar{\xi},\bar{\xi}]$ with marginal measure $\mathbb{P}_\xi$.
For any $\epsilon>0$, define $\mathfrak{B}_{\bar{\xi}}(\epsilon)
\triangleq [-\bar{\xi},-\bar{\xi}+\epsilon)\cup(\bar{\xi}-\epsilon,\bar{\xi}]$.
We say that the bound $\bar{\xi}$ is \emph{tight} with respect to functions $q,Q:\mathbb{R}_{\ge 0}\to[0,1]$ if $q(0)=Q(0)=0$, $q$ and $Q$ are non-decreasing, $q(\epsilon)>0$ for all $\epsilon>0$, and 
\begin{equation}\label{eq:coordinate_2sided}
    q(\epsilon)\ \le\ \mathbb{P}_\xi\!\left(\xi\in \mathfrak{B}_{\bar{\xi}}(\epsilon)\right)\ \le\ Q(\epsilon).
\end{equation}
\end{definition}

\begin{assumption}\label{as:PE}
The regressor sequence $\{z_k\}_{k\ge 0}$ is $(\beta,N_u,b_z)$-PE in the sense of
Definition~\ref{def:PE}.
\end{assumption}

\begin{assumption}\label{as:CDD}
There exist functions $q,Q:\mathbb{R}_{\ge 0}\to[0,1]$ such that, for each $i \in [n_x]$, the marginal distribution $\mathbb{P}_i$ of $w(i)$ is supported on $[-\bar w(i),\bar w(i)]$, and the bound $\bar w(i)$ is tight with respect to $(q,Q)$ in the sense of Definition~\ref{def:tightness}
(with $\bar{\xi}=\bar w(i)$), i.e., $q(\epsilon)\ \le\ \mathbb{P}_i\!\left(w(i)\in  \mathfrak{B}_{\bar{w}(i)}(\epsilon) \right)\ \le\ Q(\epsilon)$.
\end{assumption}

\begin{remark}
    \label{remark:Assumption34}
    Assumptions~\ref{as:PE}--\ref{as:CDD} are consistent with standard assumptions commonly used in the literature (e.g., \cite{bai1998convergence,li2024learning,akccay2004size}).
    \rev{The boundedness condition ${||z_k||}_2 \leq b_z$ also accommodates unstable systems as long as bounded regressors are ensured, for example through stabilization or data collection over multiple bounded trajectories.} 
    Assumption~\ref{as:CDD} is \revE{an element-wise} boundary-visiting condition adapted to the disturbance bound $\bar{w}$.
    Compared with global boundary-visiting condition imposed on neighborhoods of the boundary $\partial \mathcal{W}$, it is less restrictive \revE{since} it only requires non-vanishing probability near the two endpoints of each interval $[-\bar w(i),\bar w(i)]$.
    \revE{Equation~\eqref{eq:coordinate_2sided} further includes an upper-bound function $Q(\cdot)$, which quantifies the extent to which $\mathbb{P}_i$ may concentrate near the boundary points $\{-\bar{w}(i), \bar{w}(i)\}$. This inclusion is non-restrictive: since $\mathbb{P}_i$ is a probability measure, such an upper bound always exists.}
\end{remark}

This section is structured as follows. 
\rev{We first establish convergence guarantees for the data-driven feasible set under a fixed threshold $\alpha_0\in[-1,0)$, showing that the worst-case volume converges to zero almost surely.}
\rev{We then quantify how disturbance-bound mismatch perturbs the limiting feasible set, and finally bound the expected cumulative number of selected data points.}

\subsection{Feasible Set Convergence Guarantees}
\label{section:measureContraction}
\rev{Before introducing the main results in this section, we first introduce some lemmas that are required to establish the convergence of feasible set under coreset.}
Under Assumption~\ref{as:PE}, we have the  following result.
\begin{lemma}
    \label{lemma:gamma_z_lower_bound}
    Let Assumption~\ref{as:PE} be satisfied. For all $\vartheta  \in \mathbb{R}^{n_z}$ with $\left\lVert \vartheta  \right\rVert_2 = \delta > 0$, and for any $k_0 > 0$, it holds that
    \begin{equation}
        \label{eq:maxGammaz}
        \max_{k \in [k_0+1:\,k_0 + N_u]} |\vartheta ^T z_k| \geq \beta \delta.
    \end{equation}
\end{lemma}

\begin{proof}
    Pre- and post-multiplying both sides of (\ref{eq:PE}) by $\vartheta^T$ and $\vartheta$, and using $\|\vartheta\|_2 = \delta$, we obtain $\frac{1}{N_u} \sum_{k = k_0 + 1}^{k_0 + N_u} |\vartheta^T z_k|^2 \geq \beta^2 \delta^2$.
    \rev{By the pigeonhole principle, there exists at least one $k$ such that $|\vartheta^T z_k| \geq \beta \delta$, which implies (\ref{eq:maxGammaz}).}
\end{proof}

Next, we introduce a generalization of Gr\"unbaum's Inequality proposed in \cite{letwin2024generalization}.

\begin{lemma}
    \label{lemma:GrunbaumInequality}
    Let $P \subset \mathbb{R}^{n_z}$ be a convex body with centroid at origin. Let $\alpha \in (-1,0)$ and $\xi \in S^{n_z -1}$. Consider the half-space 
    \begin{equation}
        \label{eq:H_alpha^+}
        H_{\alpha}^+ = \{x \in \mathbb{R}^{n_z}: \langle x,\xi \rangle \geq \alpha h_P(-\xi) \}.
    \end{equation} 
    Then 
    \begin{equation}
        \frac{\mu({P \cap H_\alpha^+})}{\mu{(P)}} \leq C(\alpha,n_z),
    \end{equation}
    where 
    \begin{equation}\label{eq:Grunbaum_C}
        C(\alpha,n_z) = 1 - \left(\frac{n_z(\alpha + 1)}{n_z + 1}\right)^{n_z}.
    \end{equation}
\end{lemma}

This lemma allows the hyperplane to be shifted by a factor $\alpha \in (-1, 0)$ along a given direction, thus accommodating more flexible partitioning.

\begin{remark}
    In the context of the proposed coreset selection strategy, this lemma quantifies the reduction in the volume of the feasible polyhedral set after each data selection. 
    The parameter $\alpha_{+}^i$ and $\alpha_{-}^i$ in Proposition~\ref{proposition:SupportingFunctionRepresentation} correspond to the parameter $\alpha$ in (\ref{eq:H_alpha^+}).
    Therefore, the threshold condition (\ref{eq:subtrigger}) can be equivalently understood as requiring that each selection is triggered only when the feasible set undergoes a guaranteed relative volume reduction, characterized by the factor $C(\alpha_0, n_z)$.
    Proposition~\ref{proposition:SupportingFunctionRepresentation} and Lemma~\ref{lemma:GrunbaumInequality} together offer a geometric explanation of the proposed online coreset selection strategy.
\end{remark}

\begin{revBlockE}
    To quantify the uncertainty across the $n_x$ state equations, we define the worst-case volume $\mu_{\infty}(\cdot)$ as 
    \begin{equation}\label{eq:maxMeasure}
        \mu_{\infty}(\Phi(P_1,\dots, P_{n_x})) \triangleq \max_{i \in [n_x]} \mu(P_{i}),
    \end{equation}
    where $\Phi$ is the row-stacking map defined in (\ref{eq:Phi}), each $P_i$ is a polyhedral set in $\mathbb{R}^{n_z}$ and $\mu(\cdot)$ denotes the standard Lebesgue measure on $\mathbb{R}^{n_z}$, as previously used in the definition of the centroid in (\ref{eq:centroid}). 
\end{revBlockE}

Let $n^i(K)$ denote the cumulative number of selected data points caused by the $i$-th component up to time $K$, that is, 
\begin{equation}
    \label{eq:def_niK}
    n^i(K) = \sum_{k=1}^K \mathbf{1}_{\{\gamma_{k,i}=1\}},
\end{equation}
where $\mathbf{1}_{\{\cdot\}}$ denotes the indicator function.

Denote by $K^i(n)$ the time step at which the $i$-th component triggers for the $n$-th time, with the convention $K^i(0) = 0$.
Further, for any $K \geq 0$, define $\tau^i(K)$ as the first time after $K$ when the $i$-th component triggers selection, i.e.,
\begin{equation}
    \tau^i(K) \triangleq \min \{ k  > K \mid \gamma_{k ,i} = 1 \}.
\end{equation}

For a bounded polyhedral set $P$ and an inner point $\theta \in P$, we define the radius $\mathfrak{R} (P,\theta)$ as 
\begin{equation}
    \label{eq:radius}
    \mathfrak{R} (P,\theta) \triangleq \sup_{\theta' \in P} \Vert \theta - \theta' \Vert_2.
\end{equation}

Now we are in a position to state the theorem for convergence for volume of feasible set under \rev{general selection threshold}.

\begin{theorem}
    \label{Theorem:ThresholdTriggerConvergence}
    Consider system~(\ref{eq:system}) with selection threshold $\alpha_0 \in (-1,0)$. Let Assumptions~\ref{as:disturbance},\ref{as:initialSet},\ref{as:PE},\ref{as:CDD} be satisfied, and $q(\epsilon)$ be a lower bound function in (\ref{eq:coordinate_2sided}). The proposed estimator $\widehat{\Theta}_K$ has the following convergence properties.

    \begin{itemize}
        \item[(i)]
        The worst-case volume $\mu_{\infty}(\widehat{\Theta}_K)$ is bounded as:
    \rev{\begin{equation}
        \label{eq:ThresholdTrigger(i)}
        %\mu(\widehat{\Theta}_K) \leq \mu(\Theta_0)\prod_{i=1}^{n_x} C(\alpha_0,n_z)^{n^i(K)},
        \mu_{\infty}(\widehat{\Theta}_K) \leq \max_{i \in [n_x]}\{\mu(\theta_0^i) {C(\alpha_0,n_z)}^{n^i(K)}\},
    \end{equation}}
    where $C(\alpha_0,n_z)$, defined in (\ref{eq:Grunbaum_C}), is independent of $K$.            

    \item[(ii)] For all $i \in [n_x]$, with probability no less than $1-p_{\varepsilon}$, the following holds:
    \begin{equation}
        \label{eq:ThresholdTrigger(ii)}
        \frac{\tau^i(K) - K}{N_u} \leq { \ln(p_{\varepsilon})} / {\ln(1 - q\left(\frac{(-\alpha_0)\beta}{n_z}\delta \right))},
    \end{equation}
    where $p_{\varepsilon}>0$, $\delta = \mathfrak{R}(P_{K}^i,\theta_*^i)$ and $\beta$ are defined in (\ref{eq:PE}).

    \item[(iii)] The following asymptotic convergence holds almost surely: 
    \begin{equation}\label{eq:ThresholdTrigger(iii)}
        \lim_{K \to \infty} \mu_{\infty}(\widehat{\Theta}_K) = 0.
    \end{equation}
    \end{itemize}
\end{theorem}

\begin{proof}
For each triggering event associated with the $i$-th polyhedral set, by Lemma~\ref{lemma:GrunbaumInequality}, the volume of this feasible polytope is reduced by at least a factor of $C(\alpha_0,n_z)$. Since $n^i(K)$ denotes the number of triggering events for the $i$-th sub-trigger up to time $K$, it follows that 
    \rev{\begin{equation} \label{eq:proof_mu}
        \mu(P_K^i) \leq \mu(\theta_0^i) C(\alpha_0,n_z)^{n^i(K)}.
    \end{equation}} 
    Substituting the above bound in (\ref{eq:maxMeasure}) yields the desired inequality in statement (i).
    
    Fix $i \in [n_x]$ and $K > 0$. For $K_2 \geq K_1 > K$, we define event $\mathcal{E}_{[K_1:K_2]}$ as $\mathcal{E}_{[K_1:K_2]} \triangleq \{ \forall k \in [K_1:K_2], \gamma_{k,i} = 0\}$.
    To prove statement (ii), we need to derive an upper bound for $\mathbb{P}(\mathcal{E}_{[K+1:K+\Delta K]}) $.

    Consider $k \in [K+1:K+\Delta K]$. By (\ref{eq:subtrigger}), $\gamma_{k,i} = 1$ if and only if $\alpha_{+}^i \geq \alpha_0$ or $\alpha_{-}^i \geq \alpha_0$. By (\ref{eq:system}), we have $x_{k}(i) = z_{k-1}^T \theta_*^i + w_{k-1}(i)$. 
    Then, $\alpha_{+}^i \geq \alpha_0$ can be rewritten as $w_{k-1}(i) \leq -\bar{w}(i) + \psi_{k-1}^{+}$,
    where 
    \begin{equation}
        \label{eq:psi+}
        \psi_{k-1}^{+}\!=\!\left( (-\alpha_0)\widetilde{h}_{P_{k-1}^i}(z_{k-1}) + z_{k-1}^T (g(P_{k-1}^i) - \theta_*^i) \right). 
    \end{equation}
    Similarly, the condition $\alpha_{-}^i \geq \alpha_0$ holds if and only if  $w_{k-1}(i) \geq \bar{w}(i) - \psi_{k-1}^{-}$,
    where 
    \begin{equation}
        \label{eq:psi-}
        \!\psi_{k-1}^{-}\!=\!\left( (-\alpha_0)\widetilde{h}_{P_{k-1}^i}(-z_{k-1})\!-\!z_{k-1}^T (g(P_{k-1}^i) - \theta_*^i) \right). 
    \end{equation}
    Here, $\psi_{k-1}^{+}$ and $\psi_{k-1}^{-}$ are  determined by the available information at time $k-1$ (i.e., $z_{k-1}$ and $P_{k-1}^i$), and is independent of the realization of $w_{k-1}$. Also, we denote $\psi_{k} \triangleq \max\{\psi_{k}^+, \psi_{k}^{-}\}$.

    Now we assume that the event $\mathcal{E}_{[K+1:K+\Delta K]}$ holds, i.e., $\gamma_{k,i}=0$ for all $k \in [K+1:K+\Delta K]$. Then the feasible set $P_K^i$ remains unchanged during this time interval. Let $\theta^i_K = \arg\sup_{\theta \in P_K^i} \Vert \theta_*^i - \theta \Vert_2$ and define the vector $\vartheta \triangleq \theta^i_K - \theta_*^i$, we have $ \left\lVert \vartheta\right\rVert_2 = \delta $. 

    Under Assumption~\ref{as:PE}, with $\alpha_0 \in (-1,0)$,  we show that $\forall k_0 \in [K:K+\Delta K - N_u]$, the following expression holds:
    \begin{equation}
        \label{eq:max_psi}
        \max_{k \in [k_0+1:k_0 +N_u]} \psi_{k} \geq \frac{(-\alpha_0) \beta}{n_z} \delta.
    \end{equation}

    By Lemma~\ref{lemma:gamma_z_lower_bound}, denote $k' = {\arg\max}_{k \in [k_0+1:\,k_0 + N_u]} |\vartheta^T z_k|$.
    Then, $\vert \vartheta^T z_{k'} \vert \geq \beta \delta$. We first consider the case $z_{k'}^T(g(P_K^i)-\theta_*^i) \geq 0$, and distinguish two sub-cases based on the sign of $\vartheta^T z_k$. (The case for $z_{k'}^T(g(P_K^i)-\theta_*^i) < 0$ will be discussed later in a similar process)

    \begin{itemize}[leftmargin=*]
        \item If $\vartheta^T z_k \geq \beta \delta$, invoking the definition of $\widetilde{h}_{P_{K}^i}(z_{k'})$ in (\ref{eq:centeredSF}) and observing that $\theta_K^i \in P_K^i$ during the interval under consideration, we obtain
        \begin{align*}
            \psi_{k'}^{+} &= (-\alpha_0) \max_{\theta \in P_K^i}{\langle \theta - g(P_K^i),z_{k'}  \rangle} + z_{k'}^T(g(P_{K}^i)-\theta_*^i) \\
            &= (-\alpha_0)\left( \max_{\theta \in P_K^i}{\langle \theta - \theta_*^i,z_{k'}  \rangle} - z_{k'}^T( g(P_K^i)-\theta_*^i)  \right)\\
            &+ z_{k'}^T( g(P_K^i)-\theta_*^i) \\ 
            &\geq (-\alpha_0) \langle \vartheta,z_{k'} \rangle + (1+\alpha_0) z_{k'}^T( g(P_K^i)-\theta_*^i)\\ 
            & \geq \frac{(-\alpha_0) \beta}{n_z} \delta.
        \end{align*} 
        \item If $\vartheta^T z_k \leq -\beta \delta$, by (\ref{eq:h_p(pmZ)}), we have 
        \begin{equation}
            \!\!\!\max_{\theta \in P_K^i}{\langle \theta - g(P_K^i),z_{k'}  \rangle}\!\geq\!\frac{1}{n_z} \max_{\theta \in P_K^i}{\langle \theta - g(P_K^i),-z_{k'}  \rangle}.
        \end{equation}
        Then, $\psi_{k'}^{+}$ can be lower bounded by
        \begin{align*}
            \psi_{k'}^{+} &\geq \frac{(-\alpha_0) }{n_z}\max_{\theta \in P_K^i}{\langle \theta - g(P_K^i),-z_{k'}  \rangle} + z_{k'}^T(g(P_{K}^i)-\theta_*^i) \\ 
            &= \frac{(-\alpha_0) }{n_z}\left( \max_{\theta \in P_K^i}{\langle \theta - \theta_*^i,-z_{k'}  \rangle} + z_{k'}^T( g(P_K^i)-\theta_*^i)  \right)\\
            &+ z_{k'}^T( g(P_K^i)-\theta_*^i) \\ 
            &\geq \frac{(-\alpha_0) }{n_z} \langle \vartheta,-z_{k'} \rangle + \left(1+\frac{(-\alpha_0)}{n_z}\right) z_{k'}^T( g(P_K^i)-\theta_*^i) \\ 
            &\geq \frac{(-\alpha_0) \beta}{n_z} \delta.
        \end{align*}
    \end{itemize}

    In both sub-cases, the desired conclusion holds. Therefore, regardless of the sign of $\vartheta^T z_{k'}$, $\psi_{k'}^{+} \geq \frac{(-\alpha_0) \beta}{n_z} \delta$ follows.

    % When $z_{k'}^T(g(P_K^i)-\theta_*^i) < 0$, a similar analysis can be applied to $\psi_{k'}^{-}$ and we have 
\ifthenelse{\boolean{ArxivVersion}}{
    When $z_{k'}^T(g(P_K^i)-\theta_*^i) < 0$, a similar analysis can be applied to $\psi_{k'}^{-}$ and we have 
    \begin{equation}
        \label{eq:psi-bound}
        \psi_{k'}^{-} \geq \frac{(-\alpha_0) \beta}{n_z} \delta.
    \end{equation}
}
{
    When $z_{k'}^T(g(P_K^i)-\theta_*^i) < 0$, a similar analysis can be applied to $\psi_{k'}^{-}$ and we have $\psi_{k'}^{-} \geq \frac{(-\alpha_0) \beta}{n_z} \delta$.
}
    Thus, in either case, we conclude that $\max\{\psi_{k}^+, \psi_{k}^{-}\} \geq \frac{(-\alpha_0) \beta}{n_z} \delta$, 
    which prove the claim (\ref{eq:max_psi}).

    Consider partitioning the time interval $[K+1:K+\Delta K]$ into a sequence of consecutive segments of length $N_u$. Define for $t = 0, \dots, \Delta T - 1$: $\mathcal{I}_t \triangleq [K+tN_u+1:K+(t+1)N_u]$.
    If $\Delta K$ is divisible by $N_u$, define $\mathcal{I}_{\Delta T} \triangleq \emptyset$; otherwise, let $\mathcal{I}_{\Delta T} \triangleq [K+\Delta T N_u +1: K+\Delta K]$.
    Then, we have $[K+1:K+\Delta K] = \bigcup_{t=0}^{\Delta T - 1} \mathcal{I}_{t} \cup \mathcal{I}_{\Delta T}$.
    With this partition, by law of total probability, $\mathbb{P}( \mathcal{E}_{[K+1:K+\Delta K]} )$ can be written as 
    \begin{equation}
        \label{eq:ProbabilityPartition}
        \mathbb{P} (\mathcal{E}_{[K+1:K+\Delta K]}) = \prod_{t=0}^{\Delta T} \mathbb{P}(\mathcal{E}_{\mathcal{I}_t} \mid {\mathcal{E}_{[K+1:K+tN_u]}}),
    \end{equation}
    where, by convention, we write $\mathbb{P}(\mathcal{E}_{\mathcal{I}_0} \mid \mathcal{E}_{[K+1:K]}) \triangleq \mathbb{P}(\mathcal{E}_{\mathcal{I}_0})$.

    Fix $t \in [0:\Delta T -1]$, by Lemma~\ref{lemma:gamma_z_lower_bound}, there exists $k_t \in \mathcal{I}_t$ such that $\vert\vartheta^T z_{k_t}\vert \geq \beta \delta$. We denote the following interval as $\mathcal{I}_{t,1} \triangleq [K+t N_u + 1:k_t]$, 
    $\mathcal{I}_{t,2} \triangleq [k_t+1:K + (t+1)N_u]$,
    $\mathcal{I}_{t,3} \triangleq [K+1:K+tN_u]$ and 
    $\mathcal{I}_{t,4} \triangleq [K+1:k_t+1]$. 
    Therefore, $\mathbb{P}(\mathcal{E}_{\mathcal{I}_t} \mid {\mathcal{E}_{[K+1:K+tN_u]}})$ can be further written as 
    \begin{equation}
        \label{eq:P123}
        \mathbb{P}(\mathcal{E}_{\mathcal{I}_{t,1}} \mid {\mathcal{E}_{\mathcal{I}_{t,3}}}) 
        \mathbb{P}(\mathcal{E}_{k_t + 1}\mid \mathcal{E}_{[K+1:k_t]})
        \mathbb{P}(\mathcal{E}_{\mathcal{I}_{t,2}} \mid \mathcal{E}_{\mathcal{I}_{t,4}}).
    \end{equation}

    Consider the second term in (\ref{eq:P123}), $\mathcal{E}_{k_t + 1}$ holds if and only if $-\bar{w}(i) + \psi_{k_t}^{+} \leq w_{k_t} \leq \bar{w}(i) - \psi_{k_t}^{-}$.
    With (\ref{eq:max_psi}), $\mathbb{P}(\mathcal{E}_{k_t + 1}\mid \mathcal{E}_{[K+1:k_t]})$ can be written as
    \begin{align*}
        &\mathbb{P}(\mathcal{E}_{k_t + 1} \mid \psi_{k_t}^{+} \geq \frac{(-\alpha_0) \beta}{n_z} \delta ,\mathcal{F}_{k_t} ) \mathbb{P}(\psi_{k_t}^{+} \geq \frac{(-\alpha_0) \beta}{n_z} \delta \mid \mathcal{F}_{k_t})) \\ 
        +  &\mathbb{P}(\mathcal{E}_{k_t + 1} \mid \psi_{k_t}^{+} < \frac{(-\alpha_0) \beta}{n_z} \delta ,\mathcal{F}_{k_t} ) \mathbb{P}(\psi_{k_t}^{+} < \frac{(-\alpha_0) \beta}{n_z} \delta \mid \mathcal{F}_{k_t}))\\ 
        \leq &\left(1-q\left(\frac{(-\alpha_0) \beta}{n_z} \delta\right)\right) \mathbb{P}(\psi_{k_t}^{+} \geq \frac{(-\alpha_0) \beta}{n_z} \delta \mid \mathcal{F}_{k_t}))\\ 
        +  &\left(1-q\left(\frac{(-\alpha_0) \beta}{n_z} \delta\right)\right) \mathbb{P}(\psi_{k_t}^{+} < \frac{(-\alpha_0) \beta}{n_z} \delta \mid \mathcal{F}_{k_t}))\\
        = &\left(1-q\left(\frac{(-\alpha_0) \beta}{n_z} \delta\right)\right).
    \end{align*}

    Therefore, we have 
    \begin{equation}
        \mathbb{P}(\mathcal{E}_{\mathcal{I}_t} \mid {\mathcal{E}_{[K+1:K+tN_u]}}) \leq \left(1-q\left(\frac{(-\alpha_0) \beta}{n_z} \delta\right)\right).
    \end{equation}
    Then, by ($\ref{eq:ProbabilityPartition}$), it follows that:
    \begin{equation}
        \label{eq:probabilityBound}
        \mathbb{P} (\mathcal{E}_{[K+1:K+\Delta K]}) \leq \left(1-q\left(\frac{(-\alpha_0) \beta}{n_z} \delta\right)\right)^{ \Delta T }.
    \end{equation}

    Given any $p_{\varepsilon}>0$, if the right-hand side of (\ref{eq:probabilityBound}) is less than $p_{\varepsilon}$, which holds when $\Delta T$ satisfies 
    \begin{equation}
        \label{eq:DeltaT}
        \Delta T > { \ln(p_{\varepsilon})} / {\ln(1 - q\left(\frac{(-\alpha_0)\beta}{n_z}\delta \right))}.
    \end{equation}
    Inequality~(\ref{eq:DeltaT}) implies that the probability of $\mathcal{E}_{[K+1:K+\Delta K]}$ is no greater than $p_{\varepsilon}$ when $\Delta T$ exceeds the right-hand side of (\ref{eq:DeltaT}). Considering the complement event, with probability no less than $1-p_{\varepsilon}$, (\ref{eq:ThresholdTrigger(ii)}) holds.

    For statement (iii), as $K$ tends to $\infty$, we show that $\forall i \in [n_x]$, $n^i(K)$ tends to $\infty$ almost surely. 
    From previous analysis, the probability that $n^i(K)$ remains unchanged at each time step decays exponentially. 
    This implies that for any fixed $i$, an increment of one occurs for $n^i(K)$ almost surely as $K$ increases. 
    Therefore, as $K \to \infty$, $n^i(K)$ will almost surely increase without bound.
    \rev{Given $0 < C(\alpha_0,n_z) < 1$, this implies that $\lim_{K \to \infty}{C(\alpha_0,n_z)}^{n^i(K)} = 0$ for each $i$, and thus the maximum over $i$ in \eqref{eq:ThresholdTrigger(i)} also tends to $0$.}
\end{proof}

\begin{remark}\label{remark:theorem1}
    Theorem~\ref{Theorem:ThresholdTriggerConvergence} discusses both non-asymptotic and asymptotic convergence of feasible set volume under threshold trigger. 
    \rev{The statement~(i) provides an explicit upper bound on the worst-case volume $\mu_\infty(\widehat{\Theta}_K)$ in terms of the trigger counts $\{n^i(K)\}_{i\in[n_x]}$, where the contraction factor $C(\alpha_0,n_z)$ is independent of $K$.}
    The statement~(ii) characterizes a probabilistic trigger interval for next uncertainty reduction. 
    This interval highlights a trade-off controlled by the threshold $\alpha_0$: choosing $\alpha_0$ closer to $0$ typically enforces a stricter selection condition, which can lead to stronger contraction per trigger as reflected in~\eqref{eq:ThresholdTrigger(i)}, while potentially increasing the selection interval as suggested by~\eqref{eq:ThresholdTrigger(ii)}.
    The statement (iii) shows that the selection criterion~(\ref{eq:subtrigger}) retains sufficiently informative data to ensure $\mu_\infty(\widehat{\Theta}_K)$ converges to zero almost surely as $K\to\infty$.
\end{remark}

Beyond worst-case volume convergence, we also consider \textit{point-wise convergence}, i.e., the exclusion of every $\theta \neq \theta_*^i$ from the feasible set with probability one.
We observe that in the proof of Theorem~\ref{Theorem:ThresholdTriggerConvergence}, the sign of two expression $z_k^T \vartheta $ and $z_{k}^T(g(P_k^i)-\theta_*^i)$ are closely related to the classical arguments used in establishing point-wise convergence. By introducing an additional condition on the signal of two expressions, we can also derive a convergence result in the point-wise sense. The condition is formalized in the following assumption.

\begin{assumption}
    \label{as:StrongerPE}
    Under Assumption~\ref{as:PE}, there exists $N_G > 0$ such that for any $i \in [n_x]$ and $k_0 > 0$, there exists $k \in [k_0+1: k_0+N_G]$ satisfying:
    \begin{equation*}
        \label{eq:sgnPE}
        \vert z_k^T \vartheta \vert \geq \beta \delta \quad \text{and} \quad \text{sgn}(z_k^T \vartheta) = - \text{sgn}(z_{k}^T(g(P_k^i)-\theta_*^i)).
    \end{equation*} 
\end{assumption}

\begin{corollary}
    \label{corollary:convergence}
    Under assumptions of Theorem~\ref{Theorem:ThresholdTriggerConvergence} and Assumption~\ref{as:StrongerPE}, it follows that 
    \begin{equation}
        \label{eq:limThetaK}
        \lim_{K \to \infty} \widehat{\Theta}_K = \Theta_*.
    \end{equation}
\end{corollary}

\ifthenelse{\boolean{ArxivVersion}}{
\begin{proof}
    To prove (\ref{eq:limThetaK}), it suffices to prove the following statement:  For any $i \in [n_x]$, 
    \begin{equation}
        \label{eq:thetaInP_asymptotic}
        \lim_{K \to \infty} \mathbb{P}(\theta \in P_K^i) = 0
    \end{equation}
    for all $\theta \in \mathbb{R}^{n_z}$ with $\left\lVert \theta - \theta_*^i \right\rVert_2   = \delta > 0$.

    Denote $\vartheta = \theta - \theta_*^i$. Let event $\mathcal{E}_k'$ be defined as 
    \begin{equation}
        \mathcal{E}_k' \triangleq \{ w_k(i): -\bar{w}(i) - z_k^T \vartheta \leq  w_k(i) \leq \bar{w}(i) - z_k^T \vartheta \}.
    \end{equation} 
    Moreover, for $S \subseteq \mathbb{Z}_{>0}$, we define $\mathfrak{T}_{S} \triangleq \{k \in S \mid \gamma_{k,i} = 1\}$ 
    as the index set when sub-trigger $\gamma_{k,i}=1$ in set $S$. Let  $\Omega_{S} \triangleq \{ \cap_{k \in \mathfrak{T}_S } \mathcal{E}_k' \}$.
       Then, the probability of $\theta \in P_K^i$ can be written as 
       \begin{align}
           &\mathbb{P}(\theta \in P_K^i) = \mathbb{P}(\Omega_{[K]})\nonumber \\ 
           =&\left(
               \prod_{t=0}^{\Delta T - 1}\mathbb{P}(\Omega_{[tN_G + 1:(t+1)N_G]} \mid \Omega_{[tN_G]})
               \right)
               \mathbb{P}(\Omega_{[(\Delta T) N_G:K]}),\label{eq:thetaInP}
       \end{align}
       where, by convention, we write $\mathbb{P}(\Omega_{[N_G]} \mid \Omega_{0})$ as $\mathbb{P}(\Omega_{[N_G]})$ and $\Delta T = \lfloor \frac{K}{N_G} \rfloor $. 
    
    We assume $k$ to be the index satisfying (\ref{eq:sgnPE}), i.e., $u_{k}$ is designed such that 
    $\vert z_{k}^T \vartheta \vert \geq \beta \delta \quad \text{and} \quad \text{sgn}(z_{k}^T \vartheta) = - \text{sgn}(z_{k}^T(g(P_{k}^i)-\theta_*^i))$. Then, $\Omega_{k}$ implies 
    \begin{equation}
        \label{eq:event_theta_in_P}
        \gamma_{k,i} = 0 \quad  \text{or} \quad (\gamma_{k,i}=1 \text{ and } \mathcal{E}_{k}').
    \end{equation}

    We consider the case $z_{k}^T(g(P_{k}^i)-\theta_*^i) \leq 0$ (The case for $z_{k}^T(g(P_{k}^i)-\theta_*^i) \geq 0$ can be discussed similarly). By the choice of $k$, $z_{k}^T \vartheta \geq \beta \delta > 0$. From (\ref{eq:psi-bound}), we have $\psi_{k}^- \geq \frac{(-\alpha_0) \beta}{n_z} \delta$. 
    We consider the two sub-cases based on value of $\psi_{k}^-$ and $z_{k}^T \vartheta$ and show that the following implication holds:
    \begin{equation}
        \label{eq:implication0}
        \theta \in P_{k+1}^i \implies w_k(i)  \leq \bar{w}_i  - \frac{(-\alpha_0) \beta}{n_z} \delta.
    \end{equation}
    \begin{itemize}[leftmargin=*]
        \item $\psi_{k}^{-} \leq z_{k}^T \vartheta$: From (\ref{eq:subtrigger}), $\gamma_{k,i} = 1$ implies $w_k(i)\geq \bar{w}_i - \psi_{k}^-$ or $w_k(i) \leq -\bar{w}_i + \psi_{k}^+$. If $\psi_{k}^i \leq z_{k}^T \vartheta$,  $w_k(i)\geq \bar{w}_i - \psi_{k}^-$ contradicts with $\mathcal{E}_{k}'$. Combining (\ref{eq:event_theta_in_P}), we have $\theta \in P_{k+1}^i \implies \gamma_{k,i} = 0$, which can be equivalently formulated as $-\bar{w}_i + \psi_{k}^+ \leq w_k(i)  \leq \bar{w}_i - \psi_{k}^- $. Recall that $\psi_k^-$ is lower bounded by (\ref{eq:psi-bound}), we have (\ref{eq:implication0}) holds.

        \item $\psi_{k}^{-} > z_{k}^T \vartheta$: In this case, $\gamma_{k,i} = 1$ and $\mathcal{E}_{k}'$ together imply $ \bar{w}(i) - \psi_{k}^{-} \leq w_k(i) \leq \bar{w}(i) - z_{k}^T \vartheta$ or $w_k(i) \leq -\bar{w}(i) + \psi_{k}^{+}$. The event $\gamma_{k,i}=0$ implies $-\bar{w}(i) + \psi_{k}^{+} \leq w_k(i)\leq\bar{w}(i) - \psi_{k}^{-} $. Thus, $w_k(i)$ satisfies $w_k(i)\leq \bar{w}(i) - z_{k}^T \vartheta$. Recall that $z_{k}^T \vartheta > \beta \delta > \frac{(-\alpha_0) \beta}{n_z} \delta $, we conclude that (\ref{eq:implication0}) holds.
    \end{itemize}

    The above discussion shows that for any index $k$ satisfying the condition in (\ref{eq:sgnPE}), the event $\theta \in P_{k+1}^i$ implies that $w_k(i)$ satisfy $w_k(i)  \leq \bar{w}_i  - \frac{(-\alpha_0) \beta}{n_z} \delta$ or $w_k(i) \geq -\bar{w}_i + \frac{(-\alpha_0) \beta}{n_z} \delta$. Since $w_k(i)$ is independent of $z_k$ and the information $\mathcal{F}_k$ available at step $k$, the conditional probability of this event given $\mathcal{F}_k$ has the following expression
    \begin{equation}
        \mathbb{P}(\Omega_{k} \mid \mathcal{F}_k) \leq 1-q\left(\frac{(-\alpha_0) \beta}{n_z} \delta\right).
    \end{equation}

    Now we derive upper-bound for each $\mathbb{P}(\Omega_{[tN_G + 1:(t+1)N_G]} \mid \Omega_{[tN_G]})$. Under Assumption~\ref{as:StrongerPE}, there exist $k \in [tN_G + 1:(t+1)N_G]$ such that (\ref{eq:sgnPE}) holds. Then, we have 
    \begin{align*}
        &\mathbb{P}(\Omega_{[tN_G + 1:(t+1)N_G]} \mid \Omega_{[tN_G]}) \\ 
        =& \mathbb{P}(\Omega_{[tN_G : k-1]}\mid \Omega_{[tN_G]}) 
        \mathbb{P}(\Omega_{k}\mid \Omega_{[k-1]})\\
        &\mathbb{P}(\Omega_{[k+1:(t+1)N_G]}\mid \Omega_{[k]})\\
        \leq& \mathbb{P}(\Omega_{k}\mid \Omega_{[k-1]})  \leq 1-q\left(\frac{(-\alpha_0) \beta}{n_z} \delta\right).
    \end{align*}
    Therefore, the probability of $\theta \in P_k^i$ derived in (\ref{eq:thetaInP}) can be upper-bounded by 
    \begin{equation}
        \label{eq:thetaInP_Upperbound}
        \mathbb{P}(\theta \in P_K^i) \leq \left( 1-q\left(\frac{(-\alpha_0) \beta}{n_z} \delta\right) \right)^{\lfloor \frac{K}{N_G} \rfloor}.
    \end{equation}
    From (\ref{eq:thetaInP_Upperbound}), (\ref{eq:thetaInP_asymptotic}) holds as the probability exponentially decreases to 0 as $K \to \infty$, which proves the Corollary~\ref{corollary:convergence}.
\end{proof}
}
{
\begin{proof}
    Due to space limitation, the detailed proof is provided
in the arXiv version of this paper~\cite{li2025online}.
\end{proof}
}

\begin{remark}
    Corollary~\ref{corollary:convergence} provides a point-wise convergence for feasible set. This extended result is enabled by Assumption~\ref{as:StrongerPE}, which plays a key role analogous to persistent excitation. It ensures that, infinitely often, the system input $u_k$ together with $w_k$ generates sufficiently informative signal that the signs of two expressions does not be identical all the time when $z_k$ satisfies $\vert z_k^T \vartheta \vert > \beta \delta$. The structure of Assumption~\ref{as:StrongerPE} also gives insights into input design of $u_k$. It highlights how appropriately chosen inputs $u_k$ can theoretically ensure not only reduction of uncertainty in volume, but also convergence of estimator to the true parameters.
\end{remark}

\begin{revBlock}
We now examine the boundary case where $\alpha_0 = -1$. 
In Lemma~\ref{lemma:GrunbaumInequality}, the contraction factor $C(\alpha_0, n_z)$ becomes $1$, which renders the geometric contraction argument in Theorem~\ref{Theorem:ThresholdTriggerConvergence} inapplicable. 
However, this case corresponds to a conservative selection policy where any data point that imposes a cut on the current feasible set is retained. 
This leads to the following result.
\end{revBlock}

\begin{revBlock}
\begin{corollary}
    \label{Theorem:ImmediateTrigger}
    Consider system (\ref{eq:system}) under selection threshold $\alpha_0 = -1$. 
    Suppose that Assumptions~\ref{as:disturbance}--\ref{as:CDD} hold.
    Then, the estimator $\widehat{\Theta}_K$ satisfies the following convergence property almost surely:
    \begin{equation}\label{eq:corollary_immediate_trigger}
        \lim_{K \to \infty} \widehat{\Theta}_K = \Theta_*.
    \end{equation}
\end{corollary}

\begin{proof}
When $\alpha_0 = -1$, the coreset contains all informative constraints, hence $\widehat{\Theta}_K = \widehat{\Theta }_{K}^{\text{full}}$.
Under Assumptions~1--4, the almost sure convergence
$\widehat{\Theta}_K^{\mathrm{full}}\to \Theta_*$ is a classical result in set-membership identification~\cite{lu2021robust}, which implies \eqref{eq:corollary_immediate_trigger}.
\end{proof}
\end{revBlock}

\begin{revBlockG}
\subsection{Robustness to Disturbance-Bound Mismatch}
\label{section:BoundMismatch}

In this section we discuss the convergence properties when the exact upper bound $\bar{w}$ in Assumption~\ref{as:CDD} is unknown. Let $\hat{w} \triangleq \bar{w} + \Delta w$ where $\Delta w > 0$. The hyperplane constraints in \eqref{eq:H+} and \eqref{eq:H-} becomes:
\begin{align}
    H_{+}^i & = \{\theta: \quad (z_{K-1})^T \theta \leq x_{K}(i) + \hat{w}(i)\},\\
    H_{-}^i & = \{\theta: (-z_{K-1})^T \theta \leq -x_{K}(i) + \hat{w}(i)\} .
\end{align}

Let $\bar{P}_K^i$ denote the feasible set constructed under threshold $\alpha_0 = -1$ over time horizon $[K]$. We have the following lemma characterizing $\bar{P}_K^i$ under full dataset:

\begin{lemma}
    \label{lemma:untight_bound_full}
    Under Assumptions~\ref{as:disturbance}--\ref{as:CDD} and selection threshold $\alpha_0 = -1$, the $\lim_{k \to \infty }\bar{P}_k^i$ exists, defined as $\bar{P}_\infty^i$ and has the following asymptotic property almost surely:
    \begin{equation}\label{eq:P_inf_untight}
        \bar{P}_\infty^i \subseteq \theta_*^i \oplus \mathcal{B}\!\left(\mathbf{0},\frac{\Delta w(i)}{\beta}\right),
    \end{equation}
\end{lemma}

\begin{proof}
    It is sufficient to show that
    $\lim_{K \to \infty} \mathbb{P}(\theta \in \bar{P}_K^i) = 0$
    for all $\theta \in \mathbb{R}^{n_z}$ with $\left\lVert \theta - \theta_*^i \right\rVert_2   = \delta + \frac{\Delta w(i)}{\beta}$ and $\delta>0$. The event $\{\theta \in \bar{P}_K^i\}$ implies that 
    \begin{equation}\label{eq:untight_equiv}
        \{-\bar{w}(i)-\Delta w(i) - z_k^T \vartheta \leq  w_k(i) \leq \bar{w}(i) + \Delta w(i) - z_k^T \vartheta\}
    \end{equation}
    holds for all $k \in [K]$. 
    By Lemma~\ref{lemma:gamma_z_lower_bound}, we have $\max_{k \in [k_0+1:\,k_0 + N_u]} |\vartheta ^T z_k| \geq \beta \delta + \Delta w(i).$
    Then, \eqref{eq:untight_equiv} implies that there exists $k \in [k_0+1:\,k_0 + N_u]$ such that $ \vert w_k(i) \vert \leq \bar{w}(i) - \beta \delta$, which decrease $\mathbb{P}(\theta \in \bar{P}_K^i) $ by multiplying $1-q(\beta \delta)$. Thus, $\lim_{K \to \infty} \mathbb{P}(\theta \in \bar{P}_K^i) = 0$ and \eqref{eq:P_inf_untight} holds.
\end{proof}

Lemma~\ref{lemma:untight_bound_full} shows that under unknown upperbound mismatch $\Delta w$, the feasible set converge to a set $\mathcal{B}(\theta_*^i, \Delta w(i)/\beta)$ covering true parameter $\theta_*^i$ instead of converging to the point value. Now we examine how threshold selection mechanism ($\alpha_0 \in (-1,0)$) influences the convergence behavior.

To relate constraint violations to set distances, we recall the following classical result of Hoffman on polyhedral sets~\cite{hoffman2003approximate}.

\begin{lemma}
    \label{lemma:Hoffman}
    For non-empty polyhedron $P = P_{\mathcal{H}}(H,C)$, there exists a \textit{Hoffman constant} $\mathfrak{H}(P)$ such that for all $\theta \in \mathbb{R}^{n_z}$, 
    \begin{equation}
        \mathrm{dist}(\theta,P) \leq \mathfrak{H}(P) \cdot \left\lVert (H \theta - C)_{+} \right\rVert_{\infty},
    \end{equation}
    where the residual $(H \theta - C)_{+} \triangleq \max\{0,H \theta - C\}$.
\end{lemma}

Let $d_{\mathrm{H}}(X,Y)$ denote the Hausdorff distance between non-empty subsets $X \subset \mathbb{R}^{n_z}$ and $Y\subset \mathbb{R}^{n_z}$, given by
\begin{equation}
    d_{\mathrm{H}}(X,Y) \triangleq \max \{\sup_{x \in X} \mathrm{dist}(x,Y), \sup_{y \in Y} \mathrm{dist}(y,X) \}.
\end{equation}

Let $\bar{h}_k^i$ denote the maximum of centered supporting function $\tilde{h}_P(\cdot)$ (defined in \eqref{eq:centeredSF}) for feasible set $P_{t-1}^i$ over $[k]$:
\begin{equation}
\label{eq:barh_def}
\bar{h}_k^i \triangleq
\max_{1\le t\le k}\ \max\!\left\{
\tilde h_{ P_{t-1}^i}(z_{t-1}),\ \tilde h_{ P_{t-1}^i}(-z_{t-1})
\right\}.
\end{equation}

Now we introduce the following results under $\alpha_0 \in (-1,0)$:

\begin{theorem}
    \label{theorem:dist_alpha0}
    Consider system~(\ref{eq:system}) under Assumptions~\ref{as:disturbance},\ref{as:initialSet},\ref{as:PE},\ref{as:CDD}. Fix $i \in [n_x]$. 
    Let $P_K^i$ denote the feasible set generated by the selected dataset under a fixed $\alpha_0\in(-1,0)$. Then the following statements hold.
    \begin{itemize}
        \item[(i)] The Hausdorff distance between $P_K^i$ and $\bar P_K^i$ admits the upper bound 
        \begin{equation}
        \label{eq:dist_FPS_alpha0}
         d_{\mathrm{H}}(P_K^i,\bar{P}_K^i) \leq {\mathfrak{H}}(\bar{P}_K^i) \bar{h}_K^i  (\alpha_0 + 1),
        \end{equation}
        where $\mathfrak{H}(\bar{P}_K^i)$ is a Hoffman constant associated with the feasible set $\bar{P}_K^i$ constructed using the full dataset.
        \item[(ii)] The sequence $\{P_k^i\}_{k \ge 0}$ converges to a non-empty compact set $P_\infty^i$ and the following holds almost surely:
        \begin{equation}\label{eq:P_infty_untight_alpha0}
        P_{\infty}^i \!\subseteq\! \theta_*^i \oplus \mathcal{B}\left(\!\mathbf{0},\frac{ \Delta w(i)}{\beta}\!+\!\limsup_{k\to\infty} \mathfrak{H}(\bar{P}_k^i)\bar{h}_k^i (\alpha_0+1 )\!\right)\!.
    \end{equation}
    \end{itemize}
\end{theorem}

\begin{proof}
    Since $\bar{P}_K^i \subseteq P_K^i$, $\mathrm{dist}(\theta,P_K^i) = 0$ for $\theta \in \bar{P}_K^i$. We have $d_{\mathrm{H}}(P_K^i,\bar{P}_K^i) = \sup_{\theta \in {P}_K^i} \mathrm{dist}(\theta,\bar{P}_K^i)$.
    By Lemma~\ref{lemma:Hoffman}, for $\bar{P}_K^i = P_{\mathcal{H}}(H_K^i,C_K^i)$, it holds that 
    \begin{equation}
        \mathrm{dist}(\theta,\bar{P}_K^i) \leq \mathfrak{H}(\bar{P}_K^i) \cdot \left\lVert (H_K^i \theta - C_K^i)_{+} \right\rVert_{\infty}.
    \end{equation}
    For $\theta \in {P}_K^i$, the residual $(H_K^i \theta - C_K^i)_{+}$ can only be contributed by constraints that are active in representation of $\bar{P}_K^i$ but not ${P}_K^i$.
    Each constraint satisfies $\alpha_{\pm}^i < \alpha_0$ by design of selection criterion \eqref{eq:subtrigger}. Otherwise, this constraint should be selected and will not be excluded since it constitutes the full dataset representation ${P}_K^i$. 
    Let $\mathcal{M}$ denote set of indices of those constraints. 
    For $t \in \mathcal{M}$, since $- (x_t(i)+\hat{w}) = \alpha_{+}^i \widetilde{h}_{P_{t-1}^i}(z_{t-1}) - \langle z_{t-1}, g(P_{t-1}^i) \rangle$, 
    the term $\langle z_{t-1},\theta \rangle - x_t(i)-\hat{w}(i)$ in $H_K^i \theta - C_K^i$ can be written as 
    \begin{align*}
        & \langle z_{t-1},\theta -  g(P_{t-1}^i) \rangle + \alpha_{+}^i \widetilde{h}_{P_{t-1}^i}(z_{t-1})\\ 
        \leq &  \widetilde{h}_{P_{t-1}^i}(z_{t-1}) + \alpha_{+}^i \widetilde{h}_{P_{t-1}^i}(z_{t-1}) \leq (\alpha_0 + 1) \widetilde{h}_{P_{t-1}^i}(z_{t-1}).
    \end{align*}
    Similarly, $\langle -z_{t-1},\theta \rangle + x_t(i)-\hat{w}(i)\leq (\alpha_0 + 1) \widetilde{h}_{P_{t-1}^i}(-z_{t-1})$. 
    Then we have $\left\lVert (H_K^i \theta - C_K^i)_{+} \right\rVert_{\infty} \leq (\alpha_0 + 1) \bar h_K^i$,
    which yields the inequality in \eqref{eq:dist_FPS_alpha0}.
    
    The statement (ii) follows from the nestedness and compactness of $\{P_k^i\}_{k \ge 0}$ and the full-data limit set property in Lemma~\ref{lemma:untight_bound_full}.
\end{proof}

\begin{remark}\label{remark:theorem2}
Theorem~\ref{theorem:dist_alpha0} quantifies the difference between the full-dataset feasible set $\bar{P}_K^i$ and the coreset-induced feasible set ${P}_K^i$. Their Hausdorff distance is upper bounded by ${\mathfrak{H}}(\bar{P}_K^i) \bar{h}_K^i  (\alpha_0 + 1)$, where $\alpha_0=-1$ recovers the full dataset case, i.e., $P_K^i=\bar{P}_K^i$.
This also justifies the term \emph{coreset selection}: the selected inequalities form a coreset whose worst-case approximation error is explicitly controlled by $(\alpha_0+1)$.
Moreover, \eqref{eq:P_infty_untight_alpha0} indicates that the limiting estimation error can be decomposed into an irreducible component $\Delta w(i)/\beta$ due to the disturbance-bound mismatch and an additional component
$\limsup_{k\to\infty}\mathfrak H(\bar P_k^i)\bar h_k^i(\alpha_0+1)$, which is caused by discarding constraints by selection mechanism.
\end{remark}

Now we derive the upper bound on the Hoffman constant $\mathfrak{H}(\bar{P}_K^i)$ and $\bar{h}_k^i$ in the following proposition:

\begin{proposition}\label{proposition:upperbound_H}
    Let $(\tilde H_k^i,\tilde C_k^i)$ denote a minimal $\mathcal{H}$-representation of $\bar P_k^i$ (i.e., $\bar P_k^i=\{\theta:\tilde H_k^i\theta\le \tilde C_k^i\}$ and no inequality is redundant). 
    Assume that there exists $\sigma_\star>0$ such that for every $k$ and every projection point $\theta^\Pi \in \Pi_{P_k^i}(\mathbb{R}^{n_z})$, there exists an index set $I(\theta^{\Pi})$ of active constraints at $\theta^{\Pi}$ with $|I(\theta^{\Pi})|\leq n_z$ and $\sigma_{\min}\!\big((\tilde H_k^i)_{I(\theta^\Pi)}\big)\ \ge\ \sigma_\star$.
    Define $R_0^i \triangleq \mathfrak{R} (\theta_0^i,\theta_*^i)$, we have:
    \begin{equation}\label{eq:limsupH}
        \limsup_{k\to\infty} \mathfrak{H}(\bar{P}_k^i)\bar{h}_k^i \leq 2 \frac{\sqrt{n_z}}{\sigma_\star}  b_z R_0^i
    \end{equation}
\end{proposition}

\begin{proof}
    The upper bound of Hoffman constant for polyhedra can be characterized by the following inequality~\cite{klatte1995error}:
    \begin{equation}
        \mathfrak H(\bar{P}_k^i) \le \sup_{\theta^\Pi} \left\lVert \left(\tilde H_k^i (I(\theta^\Pi),:)\right)^\dagger\right\rVert_2\cdot \sqrt{|I(\theta^\Pi)|}  \le \frac{\sqrt{n_z}}{\sigma_\star},  
    \end{equation}
    where $(\cdot)^\dagger$ is the Moore--Penrose pseudoinverse and we used $|I(\theta^\Pi)|\le n_z$ and $\|A^\dagger\|_2=1/\sigma_{\min}(A)$ together with $\sigma_{\min}\!\big((\tilde H_k^i)_{I(\theta^\Pi)}\big)\ \ge\ \sigma_\star$.
    Since adding extra inequalities to an $H$-system can only increase the residual
    $\|(H\theta-C)_+\|_\infty$ while leaving $\bar P_k^i$ unchanged (redundant constraints), we have $\mathfrak H(\bar{P}_k^i) \le {\sqrt{n_z}}/{\sigma_\star}$. Regarding $\bar{h}_k^i$, it is upper bounded by $\max_{1\le t\le k} 2\left\lVert z_{t-1}\right\rVert_2 R_0^i \leq 2 b_z R_0^i$, giving \eqref{eq:limsupH}.
\end{proof}

Combining Theorem~\ref{theorem:dist_alpha0} with Proposition~\ref{proposition:upperbound_H} directly gives the following result without proof for brevity:

\begin{corollary}\label{corollary:unknownBoundAsymptoticExpression}
    Under conditions in Proposition~\ref{proposition:upperbound_H}, it holds almost surely that 
    \begin{equation}\label{eq:unknownBoundAsymptoticExpression}
        P_{\infty}^i \subseteq \theta_*^i \oplus \mathcal{B}\left(\mathbf{0},\frac{ \Delta w(i)}{\beta} +  
        c^i (\alpha_0+1 )\right),
    \end{equation}
    where $c^i \triangleq 2 \frac{\sqrt{n_z}}{\sigma_\star}  b_z R_0^i$.
\end{corollary}

\begin{remark}\label{remark:explicitRadiusBound}
Corollary~\ref{corollary:unknownBoundAsymptoticExpression} implies that the coreset-induced term in \eqref{eq:P_infty_untight_alpha0} admits an explicit bound
$2\frac{\sqrt{n_z}}{\sigma_\star}b_zR_0^i(\alpha_0+1)$; hence, the effect of $\alpha_0$ on the worst-case asymptotic radius is linear.
When $\Delta w(i)=0$ in \eqref{eq:unknownBoundAsymptoticExpression} and
$c\triangleq 2\frac{\sqrt{n_z}}{\sigma_\star}b_z(\alpha_0+1)<1$ holds uniformly, we can restart the same argument from any time $k$ by treating $P_k^i$ as the new initial set.
This yields a geometric contraction of the radii, i.e., $R_{k+\Delta k}^i\le c\,R_k^i$ for all sufficiently large $\Delta k$, and thus $R_k^i\to 0$, implying $P_\infty^i=\{\theta_*^i\}$.
This contraction interpretation is consistent with analysis in Section~\ref{section:measureContraction}, which provides a sharp convergence characterization under the threshold-trigger mechanism.
\end{remark}
\end{revBlockG}

\subsection{Selection Statistics and Expected Coreset Size}
\label{section:coresetSize}
\begin{revBlockD}
This section investigates the statistical properties of cumulative number of selected data points $n^i(K)$ defined in \eqref{eq:def_niK}.
We note that $n^i(K)$ is governed by selection probability, whose characterization requires an explicit upper bound on $\mathfrak{R}(P_k^i,\theta_*^i)$.
To relate $\mathfrak{R}(P_k^i,\theta_*^i)$ to the feasible-set volume in \eqref{eq:ThresholdTrigger(i)}, we denote the inradius of $P$ centered at $\theta$ as
\begin{equation}
    \label{eq:inradius}
    \mathfrak{r}(P,\theta) \triangleq \sup \left\{ r \ge 0 \;\middle|\; B_2(\theta,r) \subseteq P \right\},
\end{equation}
where $B_2(\theta,r) \triangleq \{ x \in \mathbb{R}^{n_z} \mid \|x-\theta\|_2 \le r \}$. 
Then, we make the following mild condition on the shape of feasible set.

\begin{assumption}\label{as:uniform_shape}
    There exist a constant $\kappa>0$ and a sequence $\{\delta_k\}_{k\ge1}$ such that for all $k\ge1$ and all $i\in[n_x]$, 
    \begin{equation}\label{eq:assumption_uniform_shape}
        \mathbb{P}\!\left(\mathfrak{R}(P_k^i,\theta_*^i)/ \mathfrak{r}(P_k^i,\theta_*^i)> \kappa\right) \le \delta_k, \quad \sum_{k=1}^\infty \delta_k < \infty.
    \end{equation}
\end{assumption}

\begin{remark}
Assumption~\ref{as:uniform_shape} allows $P_k^i$ to be arbitrarily thin at finitely many time instants.
By the Borel--Cantelli lemma, $\mathfrak{R}(P_k^i,\theta_*^i)/ \mathfrak{r}(P_k^i,\theta_*^i) \le \kappa$ holds for all sufficiently large $k$ almost surely.
Hence, the subsequent expected bound is unaffected up to an $\mathcal{O}(1)$ transient.
\end{remark}

We introduce the following notations to formalize results on expectation of $n^i(K)$.  
Let $\bar{Q}(\varepsilon)\triangleq \min\{1,Q(\varepsilon)\}$ where $Q(\varepsilon)$ is the boundary function in Assumption~\ref{as:CDD}.
Let $v_{n_z} \triangleq \mu(B_2(\mathbf{0},1))$ denote the volume of unit Euclidean ball.
Define constants $B_i \triangleq b_z\,\kappa \left(\mu(\theta_0^i)/{v_{n_z}}\right)^{1/n_z}$, $\eta(\alpha_0) \triangleq C(\alpha_0,n_z)^{1/n_z}\in(0,1)$. 
Let $q_t \triangleq \bar Q\!\left(B_i\,\eta(\alpha_0)^t\right),  t \in \mathbb{Z}_{\geq 0}$ and $S(k) \triangleq \sum_{t=0}^{k-1} (1/{q_t})$.
The following theorem characterizes an upper-bound on expectation of $n^i(K)$ w.r.t. selection threshold $\alpha_0 \in (-1,0)$.

\begin{theorem}
    \label{theorem:expected_bounds}
    Consider system (\ref{eq:system}) under Assumptions \ref{as:disturbance}--\ref{as:CDD}, \ref{as:uniform_shape} and selection threshold $\alpha_0 \in (-1,0)$. The expectation of cumulative number of selected data points $n^i(K)$ defined in \eqref{eq:def_niK} satisfies:
    \begin{equation}
        \label{eq:niK_upperbound_general}
        \mathbb{E}\left[ n^i(K) \right] \leq \sum_{k=1}^{\infty} \min\left\{1,\frac{K}{S(k)}\right\} + \sum_{k=1}^K \delta_k.
    \end{equation}
    In particular, if $Q(\varepsilon) \leq C_Q \varepsilon^p$ for some constants $C_Q>0$ and $p>0$, then
    \begin{equation}\label{eq:niK_upperbound_Q}
        \mathbb{E}\left[ n^i(K) \right] \leq 
        \frac{\ln\left(1 + C(\rho(\alpha_0)^{-1}-1)K\right)}{\ln(\rho(\alpha_0)^{-1})} + \sum_{k=1}^K \delta_k,
    \end{equation}
    where $\rho(\alpha_0) = \eta(\alpha_0)^p= C(\alpha_0,n_z)^{p/n_z} \in (0,1)$ and $C = C_Q B_i^p$.
\end{theorem}

\begin{proof}
    We proceed with the proof of~\eqref{eq:niK_upperbound_general} in two steps. 
    In the first step, we establish an upper-bound on conditional selection probability $\mathbb{P} (\gamma_{k,i}=1\mid\mathcal F_{k-1})$ w.r.t. $\mathfrak{R}(P_{k-1}^i,\theta_*^i)$. 
    From the proof of Theorem~\ref{Theorem:ThresholdTriggerConvergence}, $\gamma_{k,i}=1$ if and only if 
    \begin{equation}\label{eq:gammaki_equal_1}
        w_{k-1}(i) \leq -\bar{w}(i) + \psi_{k-1}^{+} \text{ or } w_{k-1}(i) \geq \bar{w}(i) - \psi_{k-1}^{-},
    \end{equation}
    where $\psi_{k-1}^{+}$ and $\psi_{k-1}^{-}$ are defined in \eqref{eq:psi+} and \eqref{eq:psi-}.
    Since $\left\lVert z_k \right\rVert_2 \leq b_z $ and $\max \left\lVert \theta-\theta^*\right\rVert_2 \leq \mathfrak{R}(P_{k-1}^i,\theta_*^i)$, we have 
    \begin{align*}
        &\psi_{k-1}^{+}  = \left( (-\alpha_0)\widetilde{h}_{P_{k-1}^i}(z_{k-1}) + z_{k-1}^T (g(P_{k-1}^i) - \theta_*^i) \right) =  \\ 
        & (-\alpha_0) \max_{\theta \in P_{k-1}^i}\langle z_k, \theta-\theta^* \rangle + (-\alpha_0-1) \langle z_{k-1}, \theta^* - g(P_{k-1}^i) \rangle \\
        &\leq (-\alpha_0)b_z\mathfrak{R}(P_{k-1}^i,\theta_*^i) + (1+\alpha_0) b_z\mathfrak{R}(P_{k-1}^i,\theta_*^i) \\ 
        &= b_z\mathfrak{R}(P_{k-1}^i,\theta_*^i).
    \end{align*}
    Similarly, $\psi_{k-1}^{-} \leq b_z\mathfrak{R}(P_{k-1}^i,\theta_*^i) $. Then, combining \eqref{eq:gammaki_equal_1} and Assumption~\ref{as:CDD} yields the following expression:
    \begin{equation}\label{eq:cond_gamma_bound}
        \mathbb{P} (\gamma_{k,i}=1\mid\mathcal F_{k-1}) \leq \bar{Q}(b_z\mathfrak{R}(P_{k-1}^i,\theta_*^i)).
    \end{equation}

    With the upper-bound on conditional selection probability, we proceed to the second step of the
    proof where we show that $\mathbb{E}\left[ n^i(K) \right]$ is bounded by~\eqref{eq:niK_upperbound_general}. 

    Define event $\mathcal G_k^i \triangleq \left\{\mathfrak R(P_k^i,\theta_*^i)/\mathfrak r(P_k^i,\theta_*^i) \le \kappa\right\}$. 
    By \eqref{eq:assumption_uniform_shape}, $\mathbb P((\mathcal G_k^i)^c)\le \delta_k$ where $(\cdot)^c$ denote the complement of event.
    Let $N_k\triangleq n^i(k)$ and $\Delta_k\triangleq N_k-N_{k-1}\in\{0,1\}$.
    Consider the process of $N_k^{g} \triangleq \sum_{t=1}^k \mathbf 1\{\gamma_{t,i}=1\}\mathbf 1\{\mathcal G_{t-1}^i\}$, and $\Delta_k^{g}\triangleq N_k^{g}-N_{k-1}^{g}\in\{0,1\}$.
    Then for each $k\ge1$, $\mathbf 1\{\gamma_{k,i}=1\} \le \mathbf 1\{\gamma_{k,i}=1\}\mathbf 1\{\mathcal G_{k-1}^i\} + \mathbf 1\{(\mathcal G_{k-1}^i)^c\}$, hence $N_K \le N_K^{g} + \sum_{k=1}^K \mathbf 1\{(\mathcal G_{k-1}^i)^c\}$.
    Taking expectations and using $\mathbb P((\mathcal G_{k-1}^i)^c)\le \delta_{k-1}$, we have 
    \begin{equation}\label{eq:ENK}
        \mathbb E[N_K] \le \mathbb E[N_K^{g}] + \sum_{k=1}^{K-1} \delta_k.  
    \end{equation}

    When $G_{k-1}^i$ holds, we have $\mu(P_k^i) \geq v_{n_z}\,\mathfrak r(P_k^i,\theta_*^i)^{n_z} \geq v_{n_z}\,(\mathfrak R(P_k^i,\theta_*^i)/\kappa)^{n_z}$. Then we have 
    \begin{equation}\label{eq:R_upperbound}
        \mathfrak{R}(P_K^i,\theta_*^i) \leq \kappa \left( \frac{\mu(P_K^i)}{v_{n_z}}\right)^{1/{n_z}}.
    \end{equation}
    By~\eqref{eq:R_upperbound} and \eqref{eq:proof_mu}, since $N_{k-1}\ge N_{k-1}^g$ and $\eta(\alpha_0)\in(0,1)$, 
    \begin{equation}\label{eq:bRz}
        \!\!\!b_z\,\mathfrak R(\!P_{k-1}^i,\theta_*^i)
        \!\le\!
        b_z\,\kappa\left(\frac{\mu(P_{k-1}^i)}{v_{n_z}}\right)^{1/n_z}
        \!\le\!
        B_i\,\eta(\alpha_0)^{\,N_{k-1}^{g}} .
    \end{equation} 
    Since $\bar{Q}$ is nondecreasing, \eqref{eq:cond_gamma_bound} implies
    \begin{align}
        \mathbb P(\Delta_k^{g}=1\mid\mathcal F_{k-1})
        &= \mathbf 1\{\mathcal G_{k-1}^i\}\,\mathbb P(\gamma_{k,i}=1\mid\mathcal F_{k-1}) \label{eq:delta^g1} \\ 
        &\le \bar Q\!\left(B_i\,\eta(\alpha_0)^{\,N_{k-1}^{g}}\right)
         = q_{\,N_{k-1}^{g}}.\label{eq:delta^g2}
    \end{align} 
    Define $F:\mathbb Z_{\ge0}\to\mathbb R_{\ge0}$ by $F(0)=0, F(m)=\sum_{t=0}^{m-1}\frac{1}{q_t}=S(m)\quad (t\geq 1)$. 
    Then $F(N_k^{g})-F(N_{k-1}^{g})=\mathbf 1\{\Delta_k^{g}=1\}/ q_{N_{k-1}^{g}}$ and hence $\mathbb E\!\left[F(N_k^{g})-F(N_{k-1}^{g})\mid\mathcal F_{k-1}\right]
        = {\mathbb P(\Delta_k^{g}=1\mid\mathcal F_{k-1}) / q_{N_{k-1}^{g}}} \le 1$.
    Taking expectations and summing over $k \in [K]$ yields
    \begin{equation}\label{eq:EFN}
        \mathbb E \left[ F(N_K)\right] \leq K.
    \end{equation}
    Since $F$ is non-decreasing, Markov's inequality with \eqref{eq:EFN} gives
    \begin{align*}
        \mathbb E\left[N_K^g\right] & = \sum_{k=1}^{\infty} \mathbb{P}\{N_K^g \geq k\} = \sum_{k=1}^{\infty} \mathbb{P}\{F(N_K^g) \geq F(k)\} \\ 
        &  \leq  \sum_{k=1}^{\infty}  \min \left\{1,\frac{\mathbb E\left[F(N_K^g)\right]}{F(k)}\right\} \leq \sum_{k=1}^{\infty} \min\left\{1,\frac{K}{S(k)}\right\}.
    \end{align*}
    Combining this with \eqref{eq:ENK} proves \eqref{eq:niK_upperbound_general}.

    When there exists constant $C_Q>0$ and $p>0$ such that $Q(\varepsilon) \leq C_Q \varepsilon^p$ for $\varepsilon>0$, \eqref{eq:delta^g2} becomes:
    \begin{equation*}
        \mathbb P(\Delta_k^{g}=1\mid\mathcal F_{k-1})
        \le
    C_Q\left(B_i\,\eta(\alpha_0)^{N_{k-1}^{g}}\right)^p
    =
    C\,\rho^{\,N_{k-1}^{g}}.
    \end{equation*}
    Let $\rho \triangleq \rho(\alpha_0)$. Then $\rho^{-N_k^g}=\rho^{-(N_{k-1}^g+\Delta_k^g)}=\rho^{-N_{k-1}^g}\rho^{-\Delta_k^g}$, 
    and conditioning on $\mathcal{F}_{k-1}$ gives
    \begin{align*}
        \mathbb{E}\left[\rho^{-N_k^g}\mid\mathcal F_{k-1}\right] &= \rho^{-N_{k-1}^g}\left(1 + (\rho^{-1}-1)\mathbb{E}\left[\Delta_k^g\mid\mathcal{F}_{k-1}\right]\right)\\ 
        %& \leq \rho^{-N_{k-1}^g}\left(1 + (\rho^{-1}-1)C\rho^{N_{k-1}^g}\right) \\ 
        & \leq \rho^{-N_{k-1}^g} + C(\rho^{-1}-1).
    \end{align*}
    Taking expectations and summing over $k \in [K]$ yields:
    $\mathbb{E}\left[\rho^{-N_k^g}\right] \leq 1 +  C(\rho^{-1}-1)K$.
    Since $\rho^{-x}$ is convex for $\rho \in (0,1)$, we have $\rho^{-\mathbb E[N_K^g]}\le \mathbb E[\rho^{-N_K^g}] \le 1 + C(\rho^{-1}-1)K$.
    By taking logarithms and rearranging, we have \eqref{eq:niK_upperbound_Q}.
\end{proof}

\begin{remark}\label{remark:theorem3}
    When $K > 1/C$, the right-hand side of~\eqref{eq:niK_upperbound_Q} is a non-increasing function of $\alpha_0 \in (-1,0)$. 
    Such monotonicity indicates that a larger value of $\alpha_0$ leads to a smaller expected cumulative number of selected data points, reflecting a more conservative selection policy.
    Moreover, the bound in~\eqref{eq:niK_upperbound_Q} implies that $\mathbb{E} \left[n^i(K)\right]$ grows at most on the order of $\mathcal{O}(\ln K)$.
    In contrast to the linear growth $\mathcal{O}(K)$ when no data selection and polyhedral update are performed, this logarithmic scaling ensures that the computational complexity of set-membership identification remains controlled in online learning settings.
\end{remark} 
\end{revBlockD}

\begin{revBlock}
\section{Extensions and Discussion}\label{section:ExtensionAndDiscussion}
\rev{In this section, we discuss extensions of the proposed framework to (i) nonlinear dynamics and (ii) noisy measurements, and we further analyze the computational complexity and position the method relative to least-squares identification under stochastic-noise assumptions.}

\begin{revBlockE}
\subsection{Extension to Nonlinear Cases}
\label{section:nonlinear}
We consider a class of nonlinear system $x_{k} = f(x_{k-1},u_{k-1}) + w_{k-1}$, where each component $f_i(\cdot)$ can be expressed exactly in known basis functions as $f_i(z_{k-1}) = \sum_{j=1}^{n_{\theta}} \theta^i_*(j) \phi_j(z_{k-1})$,
where $\theta^i_* \in \mathbb{R}^{n_\theta}$ is the \textit{unknown} parameter vector and $\phi(\cdot) \in\mathbb R^{n_\theta}$ is a vector of \textit{known} basis functions.
Such linear-in-the-parameters nonlinear model structures are widely used in nonlinear system identification; see, e.g.,~\cite{schoukens2019nonlinear}.
In view of~\eqref{eq:SME}, we have the following proposition:
\begin{proposition}\label{prop:nonlinearSME}
The true parameter $\theta_*^i$ lies in the set
\begin{equation}\label{eq:nonlinearSME}
    \mathcal{S}_k^i = \left\{\theta: \left\lvert \phi(z_{k-1})^T \theta - x_k(i)\right\rvert \leq \bar{w}(i)  \right\}.
\end{equation}
\end{proposition}
\begin{proof}
    From the system model, the $i$-th component satisfies $x_k(i)=f_i(z_{k-1})+w_{k-1}(i)=\phi(z_{k-1})^T\theta_*^i+w_{k-1}(i)$.
    Since $w_{k-1}(i) \in [-\bar{w}(i), \bar{w}(i)]$, $x_{k} - \sum_{j=1}^{n_{\theta}} \phi_j(z_{k-1}) \theta^i_*(j) = w_{k-1}(i) \in [-\bar{w}(i), \bar{w}(i)]$, which gives \eqref{eq:nonlinearSME}.
\end{proof}

Equation~\eqref{eq:nonlinearSME} is identical to \eqref{eq:one_state_illustration} after replacing the original regressor $z_{k-1}$ with the basis function regressor $\phi(z_{k-1})$.
Therefore, the proposed polyhedral representation~\eqref{eq:PolyhedralSet} and selection criterion~\eqref{eq:alpha+}--\eqref{eq:alpha-} remain applicable. 

Regarding convergence, the associated convergence analysis carries over to this class of nonlinear systems, provided that the regressor $\phi(z_k)$ satisfies the excitation condition analogous to that imposed in the linear case.
This yields the following corollary.

\begin{corollary}\label{corollary:nonlinear}
Under Assumption~\ref{as:disturbance},\ref{as:initialSet},\ref{as:CDD}. Assume $\{\phi(z_k)\}$ is $(\beta,N_u,b_z)$-PE (Definition~\ref{def:PE}). Then, the feasible set $\widehat{\Theta}_K$ constructed under selection threshold $\alpha_0 \in (-1,0)$ satisfies:
\begin{equation}\label{eq:nonlinearConvergence}
    \lim_{K \to \infty}\mu_{\infty}(\widehat{\Theta}_K) = 0.
\end{equation}
\end{corollary}

The proof of Corollary~\ref{corollary:nonlinear} follows the same line of argument as that of Theorem~\ref{Theorem:ThresholdTriggerConvergence} after replacing $z_k$ with $\phi(z_k)$ throughout, and is omitted for brevity.

\begin{remark}\label{remark:nonlinear}
    Although the foregoing development is presented for state equations, the representation and analysis can be extended to \emph{output regression} models.
    For instance, suppose each output component satisfies $y_k(i)=\psi(\xi_{k-1})^T \theta_*^i+\nu_{k-1}(i)$, where $\nu_{k-1}(i)$ is bounded, $\psi(\cdot)$ is a known regressor, and $\theta_*^i$ is an unknown parameter vector.
    Then the polyhedral representation and the selection criteria follow identically after replacing $x_k(i)$ and $z_{k-1}$ with $y_k(i)$ and $\psi(\xi_{k-1})$, respectively.
    Under the corresponding excitation condition on $\psi(\xi_k)$, worst-case convergence of the feasible set under data selection is guaranteed in the same sense as in Theorem~\ref{Theorem:ThresholdTriggerConvergence} and Corollary~\ref{corollary:nonlinear}.
\end{remark}

\end{revBlockE}

\subsection{Extension to Noisy Measurement}
\label{subsubsection:nonlinear}
\label{subsubsection:noisy}
Consider noisy regressor measurements $\tilde{z}_k \triangleq \begin{bmatrix}\tilde{x}_k^T & \tilde{u}_k^T\end{bmatrix}^T \in \mathbb{R}^{n_z}$, where $\tilde{z}_k=z_k + v_k$ and $v_k \in \mathcal{V}\triangleq \{v: |v|\leq \bar{v}\}$ denotes bounded measurement noise.
From (\ref{eq:system}), we have $\tilde{x}_{k}(i) - v_{k}(i) =  (\tilde{z}_{k-1} - v_{k-1})^T\theta_*^i + w_{k-1}(i)$.
Let $\hat{\zeta}_{k-1}(i) \triangleq \sum_{j=1}^{n_z}\bar{v}(j)
\sup_{\theta \in P_{k-1}^i}|\theta(j)|$,
which serves as a computable worst-case upper bound on the unknown noise-induced term
$\zeta_{k-1}(i) \triangleq v_{k-1}^\top \theta_*^i$.
Based on this, we present the following result regarding the feasibility of the true parameter.
\begin{proposition}\label{prop:noisy_robustness}
    Given that $\theta_*^i \in P_{k-1}^i$, the true parameter $\theta_*^i$ satisfies the modified constraint defined by the set $\mathcal{S}_k^i$:
    \begin{equation}
        \label{eq:noisySME}
        \mathcal{S}_k^i = \left\{\theta: 
        \left| \tilde{z}_{k-1}^\top \theta - \tilde{x}_{k}(i) \right|
        \leq
        \bar{w}(i) + \bar{v}(i) + \hat{\zeta}_{k-1}(i)
        \right\}.
    \end{equation}
\end{proposition}

\begin{proof}
    By rearranging the system dynamics, we obtain: $\tilde{z}_{k-1}^\top \theta_*^i - \tilde{x}_{k}(i) = v_{k-1}^\top \theta_*^i - w_{k-1}(i) - v_k(i)$.
    Applying the triangle inequality and the bounds on $w$ and $v$ yields: $\left| \tilde{z}_{k-1}^T \theta_*^i - \tilde{x}_{k}(i) \right| \leq \left|v_{k-1}^T \theta_*^i\right| + \bar{w}(i) + \bar{v}(i)$.
    Since $\theta_*^i \in P_{k-1}^i$, the term $\left|v_{k-1}^\top \theta_*^i\right|$ is bounded by $\hat{\zeta}_{k-1}(i)$ by definition. Consequently, $\theta_*^i \in \mathcal{S}_k^i$.
\end{proof}

Proposition~\ref{prop:noisy_robustness} implies that by replacing the fixed bound $\bar{w}(i)$ with the time-varying bound $\bar{w}(i) + \bar{v}(i) + \hat{\zeta}_{k-1}$, the linear inequality form is preserved. 
Consequently, the proposed polyhedral representation~\eqref{eq:PolyhedralSet} and triggering mechanisms~\eqref{eq:alpha+}--\eqref{eq:alpha-} remain directly applicable.

Regarding convergence, define the (unknown) tight bound on $|v_{k-1}^\top\theta_*^i|$ as $\bar{\zeta}(i)\triangleq \langle \bar v,|\theta_*^i|\rangle$. Let $\Delta \theta_0^i \in \mathbb{R}^{n_z}$ where $\Delta \theta_0^i(j) = \sup_{\theta \in \theta_0^i}\theta(j) - \inf_{\theta \in \theta_0^i}\theta(j)$. 
Then,  $0 \leq \hat{\zeta}_{k-1}(i) - \bar{\zeta}(i) \leq \langle \bar{v}, \Delta \theta_0^i \rangle \triangleq \Delta \zeta(i)$. 
The existence of $\Delta \zeta > 0$ plays the same role as an unknown disturbance-bound mismatch (cf. Section~\ref{section:BoundMismatch}). 
Then, if noisy regressor satisfies similar PE condition in Assumption~\ref{as:PE} and $v_{k-1}^\top \theta_*^i - w_{k-1}(i) - v_k(i)$ has the two-sided boundary property in Assumption~\ref{as:CDD}, the convergence analysis in Section~\ref{section:convergenceAnalysis} can be adapted, giving the following result:

\begin{corollary}\label{corollary:noisyMeasurement}
Consider system~\eqref{eq:system} under noisy regressor measurements $\tilde z_k=z_k+v_k$.
Suppose Assumptions~\ref{as:disturbance},\ref{as:initialSet} hold, and additionally:
(i) the noisy regressor $\{\tilde z_k\}$ is $(\tilde \beta,N_u,\tilde b_z)$-PE; and
(ii) for each $i$, the composite disturbance $\tilde w_{k-1}(i)\triangleq v_{k-1}^\top\theta_*^i-w_{k-1}(i)-v_k(i)$ is tight in the sense of Definition~\ref{def:tightness} with bound $\bar{\zeta}(i) + \bar{w}(i) + \bar{v}(i)$.
Then, under threshold $\alpha_0\in(-1,0)$ and conditions in Proposition~\ref{proposition:upperbound_H}, the feasible set $P_k^i$ constructed using \eqref{eq:noisySME} satisfies:
\begin{equation}\label{eq:noisyMeasurement}
P_\infty^i \subseteq  \theta_*^i \oplus \mathcal{B}\!\left(\mathbf{0},\; \frac{\Delta\zeta(i)}{\tilde\beta}
+ 2\frac{\sqrt{n_z}}{\sigma_\star} \tilde b_z\, R_0^i(\alpha_0+1)
\right).
\end{equation}
\end{corollary}

\begin{proof}
Under \eqref{eq:noisySME}, the algorithm uses the bound $\hat w_{k-1}(i)=\bar w(i)+\bar v(i)+\hat\zeta_{k-1}(i)$, whereas the (unknown) tight bound in (ii) equals $\bar w(i)+\bar v(i)+\bar{\zeta}(i)$ with $\bar{\zeta}(i)=\langle\bar v,|\theta_*^i|\rangle$. 
Since $0\le \hat\zeta_{k-1}(i)-\bar{\zeta}(i)\le \Delta\zeta(i)$, the setting reduces to the bound-mismatch case of Corollary~\ref{corollary:unknownBoundAsymptoticExpression} with $(\beta,b_z)$ replaced by $(\tilde\beta,\tilde b_z)$, yielding \eqref{eq:noisyMeasurement}.
\end{proof}

\begin{remark}\label{rem:noisy_conservatism}
Eq.\eqref{eq:noisyMeasurement} highlights two sources of conservatism in estimation error.
The first term $\Delta\zeta(i)/\tilde\beta$ is induced by the regressor measurement noise and originates from replacing the unknown tight bound $\bar{\zeta}(i)=\langle\bar v,|\theta_*^i|\rangle$ with the worst-case bound $\hat\zeta_{k-1}(i)$; this replacement yields an effective bound mismatch uniformly bounded by $\Delta\zeta(i)$.
The second term $2\frac{\sqrt{n_z}}{\sigma_\star}\tilde b_z R_0^i(\alpha_0+1)$ is the coreset-selection relaxation characterized in Section~\ref{section:BoundMismatch}, and it vanishes in the full-data limit $\alpha_0\to -1$.
Finally, when $\bar v\to 0$, we have $\Delta\zeta(i)\to 0$ and the bound reduces to the noiseless-regressor counterpart.
\end{remark}
\end{revBlock}

\begin{revBlock}
\subsection{Computational Complexity}
\label{section:computationalComplexity}
The total computational cost over the horizon $K$ for the $i$-th polyhedron is given by $\mathcal{C}(K)
    = \sum_{k=1}^K \Cselect(k)
    + \sum_{k=1}^K \mathbf{1}_{\{\gamma_{k,i}=1\}} \Cupdate(k)$,
where $\Cselect(k)$ represents the cost of checking triggering conditions, and $\Cupdate(k)$ includes the cost of redundancy removal and centroid maintenance.
Let $\mathcal{C}_{\text{LP}}(m, n_z)$ denote the polynomial time cost of solving a linear program (LP) with $m$ constraints and $n_z$ variables~\cite{boyd2004convex}.
Evaluating $\Cselect(k)$ requires solving two LPs, i.e., $\Cselect(k) = \mathcal{O}(\mathcal{C}_{\text{LP}}(m, n_z))$.
When an update is triggered ($\gamma_{k,i}=1$), LP-based redundancy removal requires solving at most $m$ LPs, yielding $\Cupdate(k) = \mathcal{O}(m\, \mathcal{C}_{\text{LP}}(m, n_z))$, which is typically larger than the linear cost of centroid maintenance~\cite{kannan2009random}.

The above analysis implies the following:
(i) When updates occur frequently (i.e., $\mathbb{P}(\gamma_{k,i}=1)>\mathcal{O}(1/m)$), the computational bottleneck is the update operations triggered by $\gamma_{k,i}=1$. 
By \eqref{eq:niK_upperbound_Q} in Theorem~\ref{theorem:expected_bounds}, the expected cumulative number of updates $\mathbb{E}\left[\sum_{k=1}^K \mathbf{1}_{\{\gamma_{k,i}=1\}}\right] = \mathbb{E}\left[ n^i(K) \right]$ grows logarithmically with $K$ and is non-increasing in $\alpha_0 \in (-1,0)$. 
Therefore, increasing $\alpha_0$ mitigates the frequency of the computationally expensive $\Cupdate(k)$ operations and thereby lowers the overall computational burden.
(ii) In the long run, the logarithmic growth of $m$ (as implied by $m \le m_0 + 2 n^i(K)$ together with~\eqref{eq:niK_upperbound_Q}) implies that the total complexity becomes dominated by the selection cost, and the per step complexity scales as $\mathcal{O}(n_x \mathcal{C}_{\text{LP}}(\ln K, n_z))$.
This logarithmic characterization guarantees the long-term tractability of the proposed framework. Moreover, since the dependence on the problem dimension enters only through the LP complexity, the proposed method remains applicable to moderate-to-high dimensional systems for which LPs of size $(m,n_z)=(\mathcal{O}(\ln K),n_z)$ can be solved within the available computational budget.

\begin{remark}[Accuracy--Complexity Trade-off]\label{remark:tradeoff_alpha0}
The threshold $\alpha_0$ is a tunable parameter governing the trade-off between identification accuracy and computational complexity.
On the accuracy side, a larger $\alpha_0$ enforces a stricter coreset selection and thus increases the coreset-induced conservativeness, as quantified by Theorem~\ref{theorem:dist_alpha0}.
On the complexity side, fewer selections lead to fewer polyhedral updates; correspondingly, Theorem~\ref{theorem:expected_bounds} shows that the expected number of updates decreases with $\alpha_0$.
Consequently, $\alpha_0$ provides a principled accuracy--complexity trade-off: choose a larger $\alpha_0$ to meet a real-time budget, and a smaller $\alpha_0$ when faster finite-time uncertainty reduction is prioritized.
\end{remark}

\begin{remark}\label{remark:reduce_complexity}
    In addition to tuning $\alpha_0$, we can further reduce the computational burden in real-time deployments by adopting several strategies.
    First, the update process need not be strictly sequential; selected data points can be buffered to perform redundancy removal and centroid updates in batches.
    Second, since the set-membership updates for each state component $x(i)$ are decoupled, the algorithm admits a parallel implementation, which can reduce the per-step execution time.
\end{remark}
\end{revBlock}

\begin{revBlock}
\subsection{Comparison with Least-Squares Criterion}
\label{section:comparison_LS}

In the convergence analysis, Assumption~\ref{as:CDD} is used to capture the stochastic richness of the disturbance. Under stochastic noise, an alternative for identifying unknown parameters is least-squares identification (LSI). In this section, we briefly compare the proposed coreset selection-based SMI with LSI, highlighting key advantages and the primary limitation.

Three main advantages are as follows:
First, the SMI estimator $\Theta_{\mathcal{S}}$ yields a \textit{deterministic feasible parameter set} rather than a point estimate with probabilistic confidence bounds as in LSI. 
Such hard guarantees on the parameter set are essential for safety-critical applications like robust model predictive control~\cite{lorenzen2019robust} and reachability analysis \cite{alanwar2023data}. 
Second, our data selection mechanism \eqref{eq:subtrigger} induces a geometric contraction of the feasible-set measure and ensures \textit{worst-case convergence} (Lemma \ref{lemma:GrunbaumInequality},  Theorem~\ref{Theorem:ThresholdTriggerConvergence}), whereas data selection in LSI typically optimizes average-case prediction accuracy~\cite{chatterjee1986influential}. This worst-case convergence is particularly beneficial for robust synthesis, where reducing the maximum admissible uncertainty is often the primary bottleneck~\cite{bemporad2007robust}.
Finally, the SMI framework remains theoretically tractable when the noise is not well characterized (e.g., systematic bias or bound mismatch). 
In particular, the proposed method provides a deterministic characterization of the uncertainty set even when the exact bounds are unknown (Theorem~\ref{theorem:dist_alpha0}). By comparison, least-squares methods typically handle these effects via robustified variants that assume a specific uncertainty model and a prescribed robustness level, adding extra assumptions and tuning beyond standard LSI~\cite{el1997robust}.

The primary drawback of SMI is its higher computational complexity relative to the fixed cost of LSI.
These limitations are mitigated by the proposed coreset selection mechanism, which reduces computational complexity (Section~\ref{section:computationalComplexity}) with convergence guarantees (Section~\ref{section:convergenceAnalysis}), thereby rendering the SMI-based framework more practical for online implementation.
\end{revBlock}

\section{Simulations}
\label{section:simulation}
In this section we provide simulation results to validate the theoretical results and show the effectiveness of the proposed online coreset selection strategy.
\rev{The following simulations are performed utilizing a laptop with Intel Core i5-11300H CPU and 16 GB of RAM, and the optimization problems are solved by MATLAB function \texttt{linprog()}.}
We first present an illustrative example on a second-order system to visualize the uncertainty evolution and selection behavior. 

\subsection{Illustrative Example: Second-Order System}
We begin by illustrating the proposed data selection method on a second-order discrete-time linear system of the form~\eqref{eq:system} with system matrices\begin{equation*}
    \label{eq:SimSys_N2}
    A = \begin{bmatrix}
        0.5366 & 0.2038 \\
        -0.0406 & 0.6310
    \end{bmatrix}, \quad
    B = \begin{bmatrix}
        -0.02 \\
        0.4730
    \end{bmatrix},
\end{equation*}
as used in~\cite{lorenzen2019robust}. The disturbance sequence $w_k$ is i.i.d. and uniformly distributed over the ellipsoid set $\mathcal{W} = \{w \in \mathbb{R}^2: \|w\|_2 \leq 0.5\}$. The trigger threshold is set as $\alpha_0 = -0.3$. 
\rev{The input sequence $u_k$ is generated i.i.d. from  $\mathcal{N}(0,5)$, and the simulation length is $K = 150$ steps.} 
We choose the initial feasible set of $\theta_0^i(i \in \{1,2\})$ as a unit cube $\{x \in \mathbb{R}^3: \Vert x \Vert_{\infty} \leq 1\}$. 
Then it has $\mathcal{H}$-representation $P_{\mathcal{H}}([I;-I],\mathbf{1}_{6})$.

\subsubsection{Verification of PE}

Fig.~\ref{fig:PE} represents the smallest eigenvalue of expression $\lambda_{\min} = (\sum z_k z_k^T)/N_u$ over sliding windows. 
Specifically, each blue marker represents the minimum eigenvalue of the matrix $(\sum_{k=k_0-N_u+1}^{k_0} z_k z_k^T)/N_u$ at time step $k_0$, while red dashed line denotes the constant threshold $\beta^2$. 
As observed, all eigenvalues remain above the threshold, thereby showing that the Assumption~\ref{as:PE} is satisfied.

\begin{figure}[htbp]
    \centering
    \includegraphics[width=.46\textwidth]{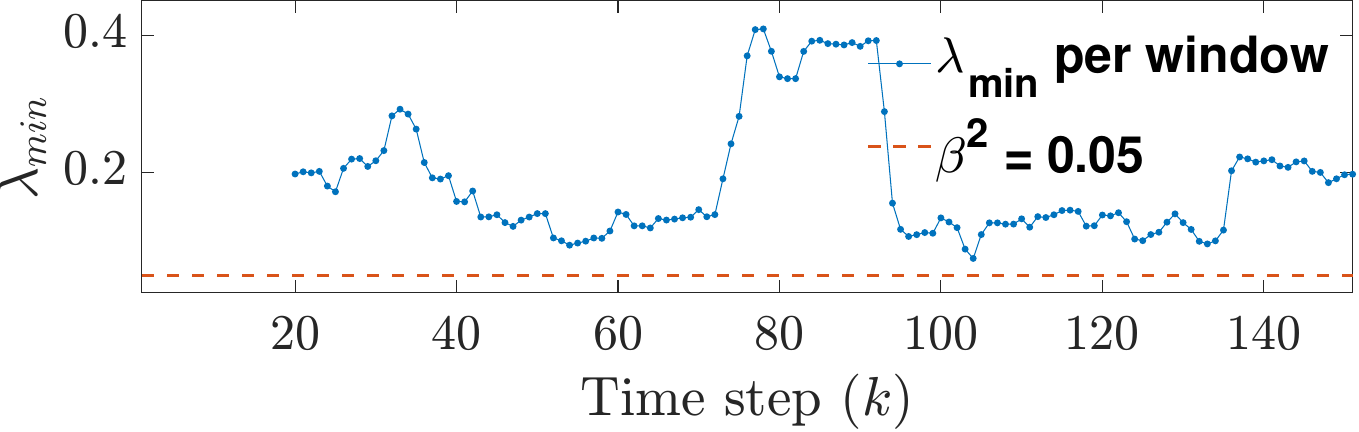}%
    \caption{Verification of PE condition in Assumption~\ref{as:PE}: minimum eigenvalue of the covariance matrix over sliding windows of length $N_u=20$.}
    \label{fig:PE}
\end{figure}

\subsubsection{Data Selection and Uncertainty Reduction}
To illustrate the effect of threshold-triggered data selection, we simulate the evolution of the feasible set over time. 
The result is shown in Fig.~\ref{fig:N2_performance} and Fig.~\ref{fig:N2_3D_evolution}.

Fig.~\ref{fig:N2_performance} shows the evolution of the worst-case  volume $\mu_{\infty}(\widehat{\Theta}_k)$ and cumulative number of selected data points $\mathcal{T}(k)$. 
From Fig.~\ref{fig:N2_performance}(a), we see that the volume of feasible set decreases steadily as more informative data is selected. 
From Fig.~\ref{fig:N2_performance}(b), we see that data selection is triggered only 25 times over the entire 150-step horizon.
Together, Figs.~\ref{fig:N2_performance}(a) and (b) illustrate a representative execution of the proposed online coreset selection method.

\begin{figure}[htbp]
    \centering
    \includegraphics[width=.43\textwidth]{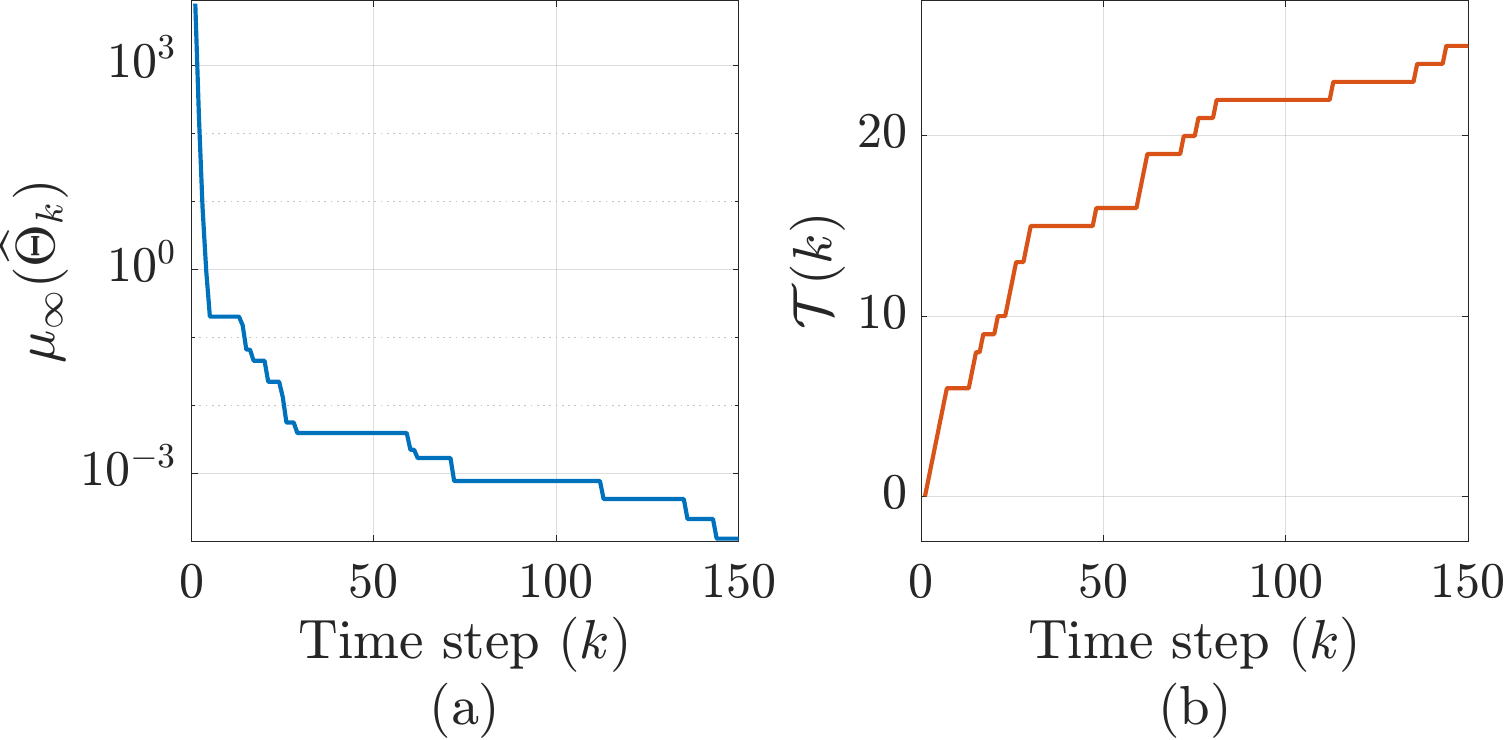}%
    \caption{Performance of the threshold-triggered estimator ($\alpha_0 = -0.3$). (a) Evolution of the worst-case volume $\mu_{\infty}(\widehat{\Theta}_k)$. (b) Cumulative number of selected data points.}
    \label{fig:N2_performance}
\end{figure}

Fig.~\ref{fig:N2_3D_evolution} visualizes the uncertainty contraction process. 
The polyhedral set $P_k^1$ and $P_k^2$ are shown at time steps $k=0$, $k=20$, and $k=80$. 
The true parameter values are marked by red stars. 
It is observed that the volume of each 3D polytope reduces significantly over time, and the sets progressively contract toward the true parameter values, validating the effectiveness of the proposed data selection strategy in selecting informative data.
\begin{figure}[htbp]
    \centering
    \includegraphics[width=.5\textwidth]{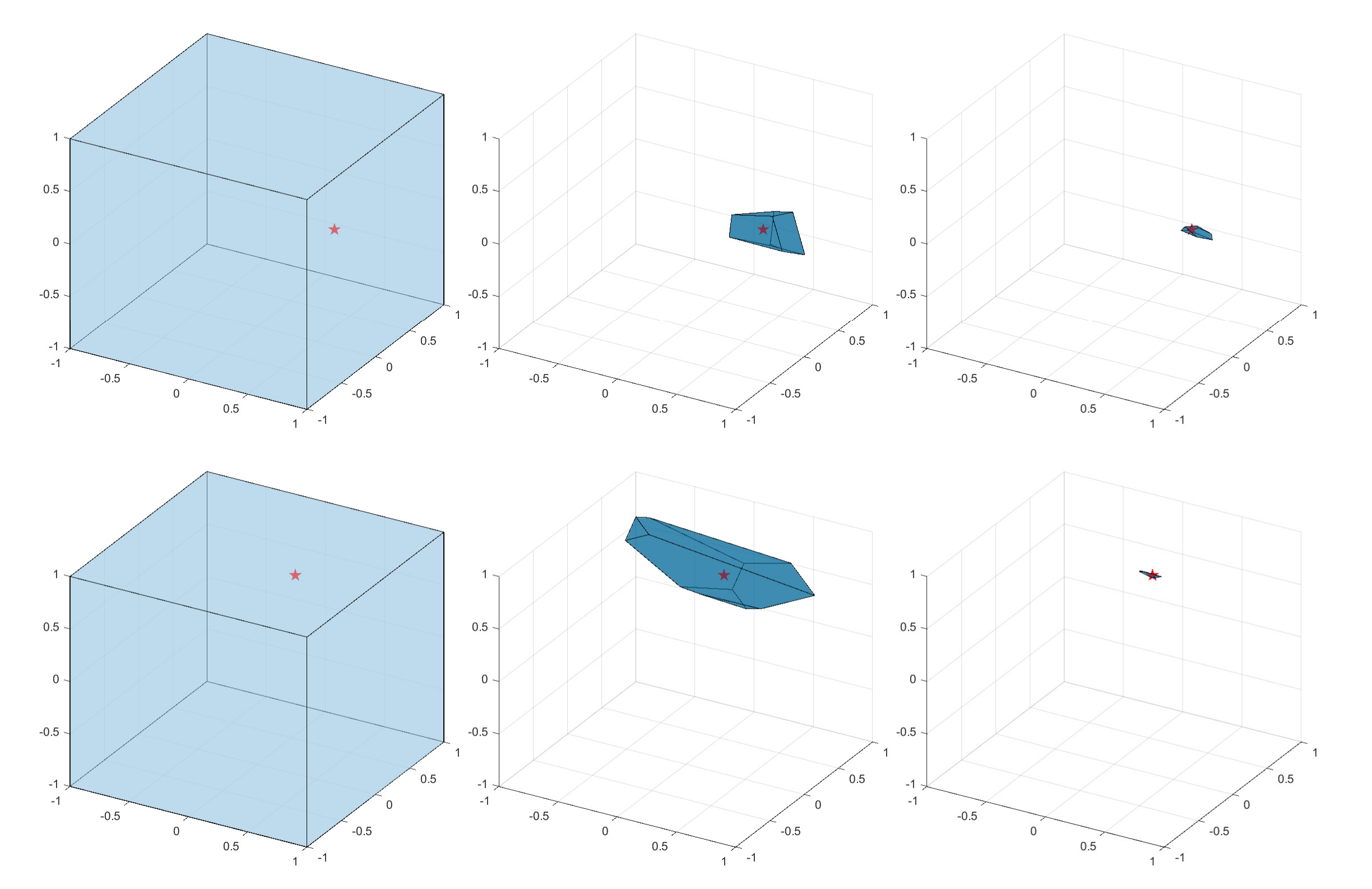}%
    \caption{Evolution of feasible parameter sets. The top and bottom rows depict the 3D polyhedral uncertainty sets $P^1_k$ and $P^2_k$ at $k = 0$, $k = 20$, and $k = 80$, respectively. Red stars mark the true parameter values.}
    \label{fig:N2_3D_evolution}
\end{figure}

\subsection{Performance Under Different Selection Threshold $\alpha_0$}

In this section we simulate the performance of the proposed coreset selection scheme over a grid of selection threshold $\alpha_0 \in \{-1,-0.8,-0.6,-0.5,-0.4,-0.3,-0.2,-0.1,0\}$.
We consider the linearized discrete-time system of Boeing 747~\cite{li2024learning} of the form (\ref{eq:system}) with the following system matrices:
\begin{equation*}
    \label{eq:simulation}
    A = \begin{bmatrix}
            0.99&0.03&-0.02&-0.32\\ 
            0.01&0.47&4.7&0 \\ 
            0.02 & -0.06 & 0.4 & 0 \\ 
            0.01 & -0.04 & 0.72 & 0.99
        \end{bmatrix}\!,\!
    B = \begin{bmatrix}
            0.01 & 0.99 \\ 
            -3.44 & 1.66 \\ 
            -0.83 & 0.44 \\ 
            -0.47 & 0.25
        \end{bmatrix}\!.
\end{equation*}
The disturbance sequence $w_k$ is i.i.d. and uniformly distributed over the hypercube $\mathcal{W} = \{w \in \mathbb{R}^4: \|w\|_{\infty} \leq 2\}$. 
\rev{The input sequence $u_k$ is generated i.i.d. from  $\mathcal{N}(\mathbf{0} ,I_2)$. The simulation horizon is set to $K=500$.}
The initial feasible set of $\theta_0^i (i \in \{1,2,3,4\})$ are chosen as a hypercube $\{x \in \mathbb{R}^6: \Vert x \Vert_{\infty} \leq 10\}$, with $\mathcal{H}$-representation $P_{\mathcal{H}}([I;-I],10\times\mathbf{1}_{12})$.
For each $\alpha_0$, we performed $N_{epoch}=100$ Monte Carlo trials with randomly generated regressor sequences ${z_k}$ satisfying the PE condition (Assumption~\ref{as:PE} with $N_u = 20$ and $\beta^2 = 0.02$).

\begin{revBlock}
    Fig.~\ref{fig:N4_alpha0}(a) and Fig.~\ref{fig:N4_alpha0}(b) illustrate the evolution of the feasible set volume and the cumulative number of selected data points, respectively. 
    For clarity, we display the results for a representative subset of $\alpha_0$ values, i.e., $\alpha_0 \in \{-1.0, -0.5, -0.3, -0.2, 0.0\}$. 
    As shown in Fig.~\ref{fig:N4_alpha0}(a), the average worst-case volume $\mu_{\infty}(\widehat{\Theta}_k)$ converges to zero for all $\alpha_0$, validating the   theoretical convergence guarenteed by Theorem~\ref{Theorem:ThresholdTriggerConvergence}. 
    Simultaneously, the cumulative number of selected data grows at a rate of $O(\ln K)$, \rev{which is}  consistent with Theorem~\ref{theorem:expected_bounds}. 
    Notably, a larger $\alpha_0$ (closer to 0) results in fewer selected points and less computationally expensive polyhedral updates, albeit at the cost of increased conservatism \revE{in the feasible set approximation}.
\begin{figure}[htbp]
    \centering
    \includegraphics[width=.49\textwidth]{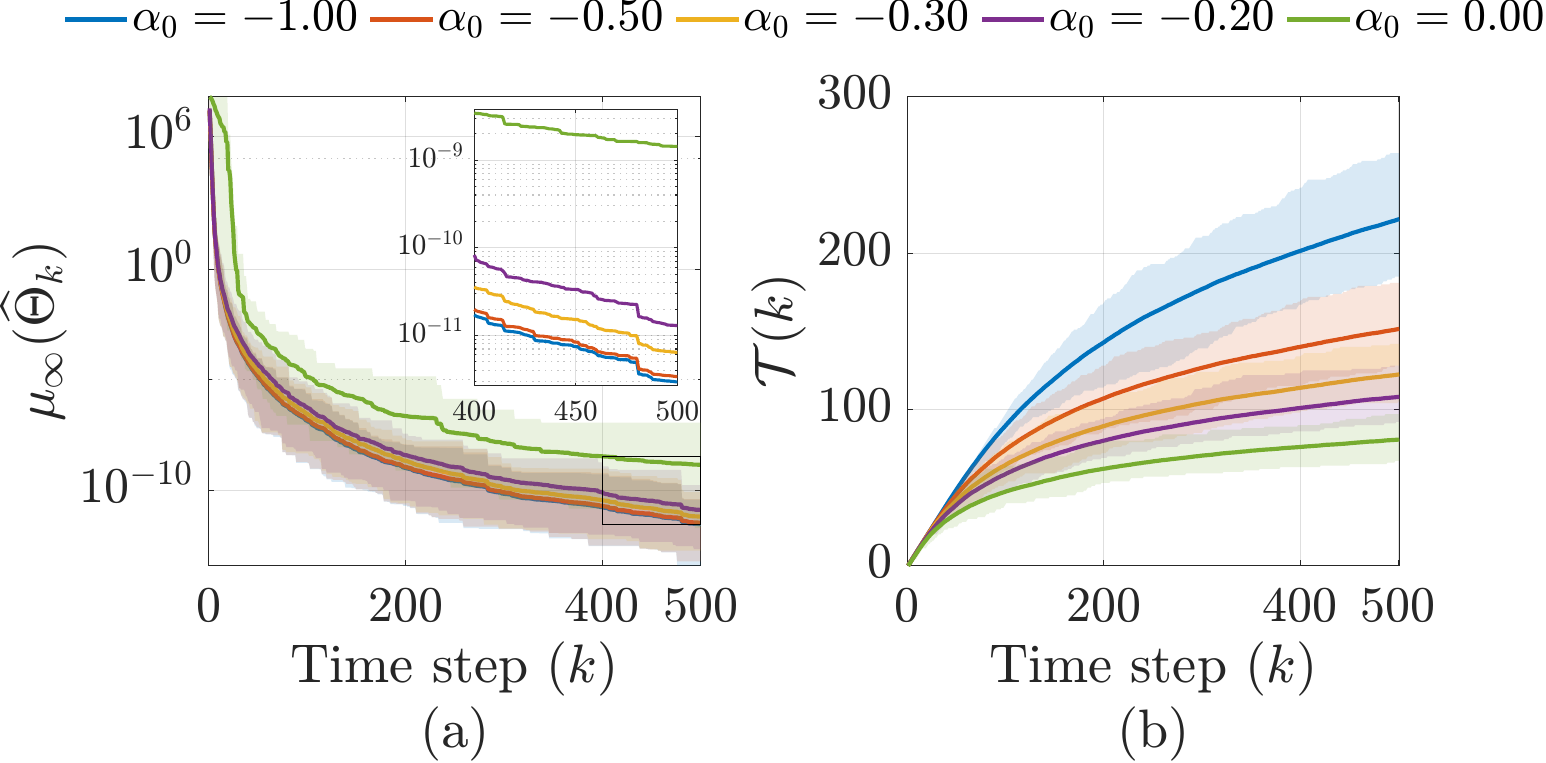}%
    \caption{Comparison of performance under different threshold parameter $\alpha_0 \in \{-1,-0.5,-0.3,-0.2,0\}$. (a) Evolution of feasible set worst-case volume $\mu_{\infty}(\widehat{\Theta}_k)$. (b) Cumulative number of selected data points under different $\alpha_0$ values. The shaded regions represent the range between the minimum and maximum values.}
    \label{fig:N4_alpha0}
\end{figure}

    To further illustrate this trade-off, Fig.~\ref{fig:tradeoff} plots the final worst-case volume $\mu_{\infty}(\widehat{\Theta}_{500})$ (left axis) and the average simulation time (right axis) against $\alpha_0$.
    This dual-axis \revE{figure} clearly demonstrates that increasing $\alpha_0$  from -1 to 0 significantly reduces the computational burden while slightly increasing the uncertainty set volume.
    Consequently, $\alpha_0$ acts as a tuning parameter to balance estimation accuracy and computational \revE{cost}.
\begin{figure}[htbp]
    \centering
    \includegraphics[width=.45\textwidth]{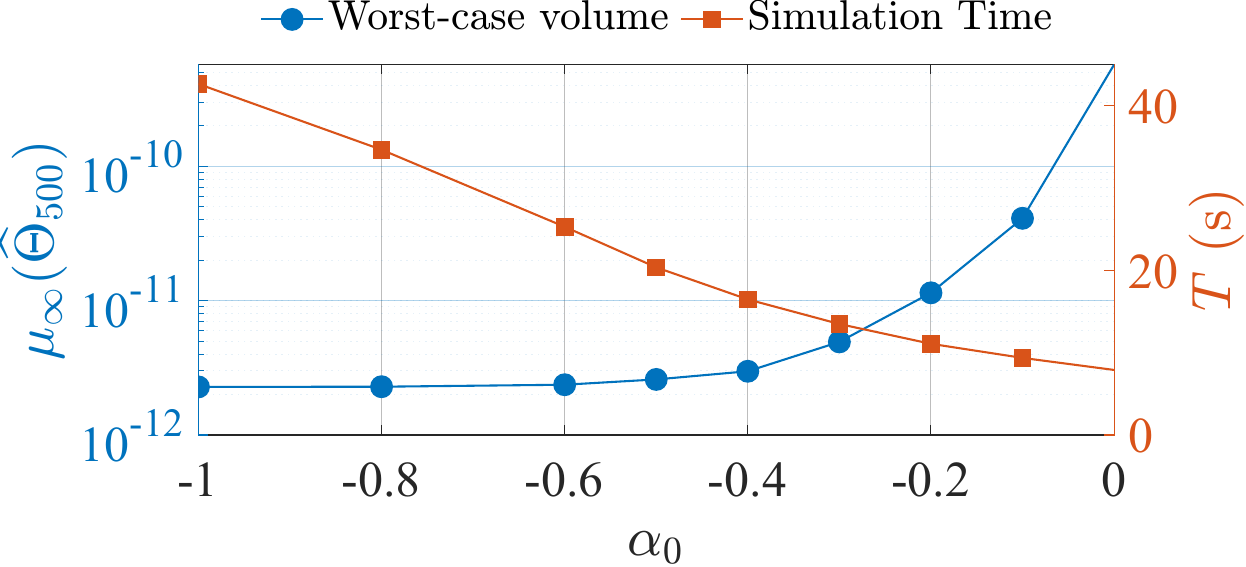}%
    \caption{Trade-off between estimation accuracy and computational efficiency with respect to $\alpha_0$. The blue curve represents the final worst-case volume $\mu_{\infty}(\widehat{\Theta}_{500})$, while the orange curve depicts the average simulation time $T$.}
    \label{fig:tradeoff}
\end{figure}

\end{revBlock}

\section{Conclusion and Future Directions}
\label{section:conclusions}
In this work, We propose an online coreset selection method for set-membership identification with guaranteed convergence. 
By deriving a stacked polyhedral over-approximation of the feasible parameter set, we establish volume convergence under persistently exciting data and tight disturbance bounds via a generalized Grünbaum-type inequality.
We further characterize the feasible set under disturbance-bound mismatch and provide an expected upper bound on the number of data points in coreset. 
We also discuss extensions to linear-in-the-parameters nonlinear dynamics and bounded measurement noise. 
Finally, we conduct numerical simulations and validate the theoretic results. 

\rev{Future work includes extending the framework to broader nonlinear dynamics e.g., via Koopman-based representations that yield an approximate linear parameter structure, together with a quantitative analysis of how approximation errors affect feasibility and convergence. 
Moreover, since the noisy-measurement extension in Section~\ref{subsubsection:noisy} preserves feasibility by relaxing the disturbance bound, deriving less conservative, data-dependent bounds that remain provably convergent under coreset selection is an important next step.}

\appendices

\section*{Acknowledgment}
\revE{The authors would like to thank the Associate Editor and the anonymous reviewers for their suggestions that have improved the quality of this work.}

\section*{References}
\bibliographystyle{IEEEtran}  
\bibliography{myref}

\begin{IEEEbiography}[{\includegraphics[width=1in,height=1.25in,clip,keepaspectratio]{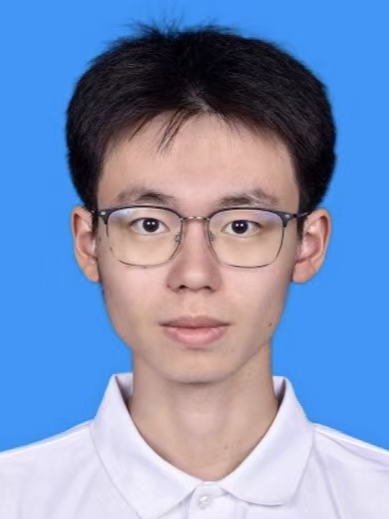}}]
    {Jingyuan Li} was born in Zhuzhou, Hunan, China. He received the B.Eng. degree in Automation from the Beijing Institute of Technology.

    He is taking successive postgraduate and doctoral programs at the School of Automation, Beijing Institute of Technology. His research interests include system identification, event-triggered learning and model predictive control.
\end{IEEEbiography}

\begin{IEEEbiography}[{\includegraphics[width=1in,height=1.25in,clip,keepaspectratio]{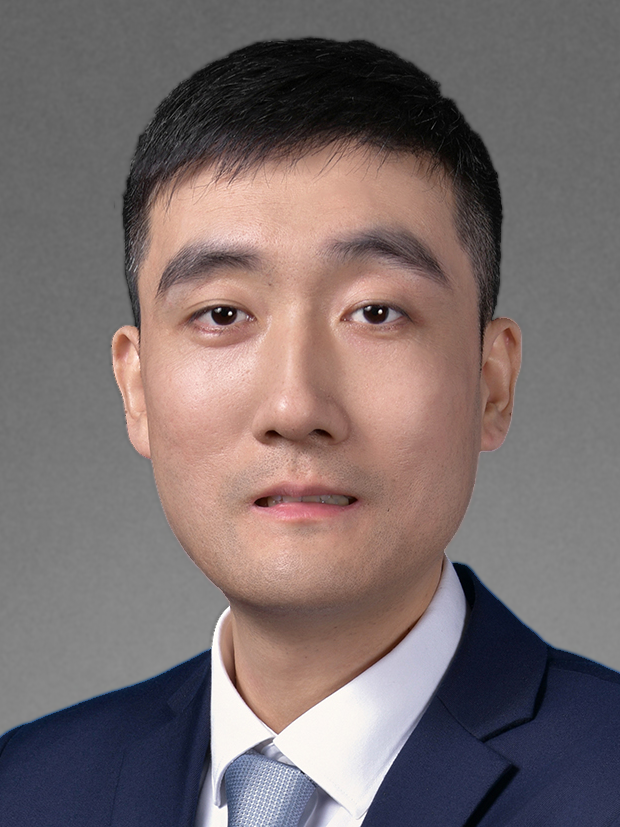}}]
{Dawei Shi} is presently a Professor with the School of Automation, Beijing Institute of Technology. He received the B.Eng. degree in Electrical Engineering and Automation from the Beijing Institute of Technology, Beijing, China, in 2008, and the Ph.D. degree in Control Systems from the University of Alberta, Edmonton, Alberta, Canada, in 2014. From 2017 to 2018, he was a Postdoctoral Fellow in Bioengineering at the Harvard University, Cambridge, MA, USA. 

Dr. Shi's research interests include the analysis and design of advanced sampled-data control systems, with applications to biomedical engineering, robotics, and motion systems. He has served as an Associate Editor/Technical Editor for several international journals, including Control Engineering Practice, IEEE/ASME Transactions on Mechatronics, IEEE Transactions on Industrial Electronics, and IEEE Control Systems Letters. He served as the General Chair of the IEEE Industrial Electronics Society Annual Online Conference in 2024 and 2025. He is a Senior Member of the IEEE and a Fellow of the IET.
\end{IEEEbiography}

\begin{IEEEbiography}[{\includegraphics[width=1in,height=1.25in,clip,keepaspectratio]{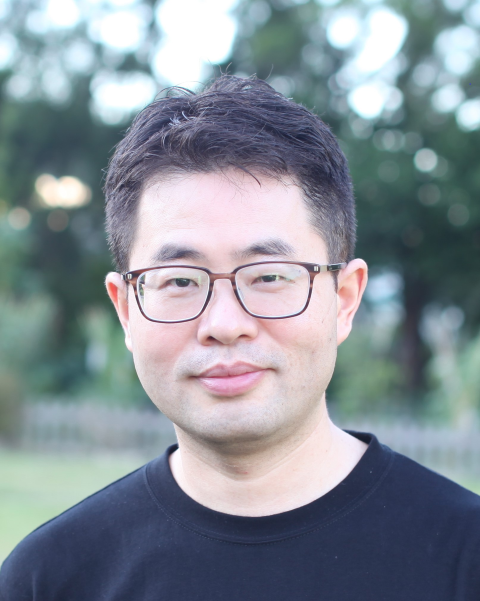}}]
{Ling Shi} received his B.E. degree in Electrical and Electronic Engineering from The Hong Kong University of Science and Technology~(HKUST) in 2002 and Ph.D. degree in Control and Dynamical Systems from The California Institute of Technology~(Caltech) in 2008. 
He is currently a Professor in the Department of Electronic and Computer Engineering at HKUST with a joint appointment in the Department of Chemical and Biological Engineering (2025-2028), and the Director of The Cheng Kar-Shun Robotics Institute~(CKSRI). 

His research interests include cyber-physical systems security, networked control systems, sensor scheduling, event-based state estimation, and multi-agent robotic systems (UAVs and UGVs). He served as an editorial board member for the European Control Conference 2013-2016. He was a subject editor for International Journal of Robust and Nonlinear Control (2015-2017), an associate editor for IEEE Transactions on Control of Network Systems (2016-2020), an associate editor for IEEE Control Systems Letters (2017-2020), and an associate editor for a special issue on Secure Control of Cyber Physical Systems in IEEE Transactions on Control of Network Systems (2015-2017). He also served as the General Chair of the 23rd International Symposium on Mathematical Theory of Networks and Systems (MTNS 2018). He is currently serving as a member of the Engineering Panel (Joint Research Schemes) of the Hong Kong Research Grants Council (RGC) (2023-2026). He received the 2024 Chen Han-Fu Award given by the Technical Committee on Control Theory, Chinese Association of Automation (TCCT, CAA). He is a member of the Young Scientists Class 2020 of the World Economic Forum (WEF), a member of The Hong Kong Young Academy of Sciences (YASHK), and he is an IEEE Fellow.
\end{IEEEbiography}

\end{document}